\documentclass{article}

\usepackage{fullpage}
\usepackage{authblk}

\usepackage[numbers,sort&compress]{natbib}

\usepackage[normalem]{ulem}

\usepackage[english=american]{csquotes}
\usepackage{graphicx,amsmath,amssymb,amsfonts,amsthm,mathtools}

\usepackage[hidelinks]{hyperref} 
\usepackage[dvipsnames]{xcolor}

\usepackage{enumerate} 
\usepackage{booktabs} 

\usepackage{longtable}

\usepackage[width=1\textwidth,labelfont=bf]{caption}

\newtheorem*{remark}{Remark}

\usepackage[numbers,sort&compress]{natbib}

\usepackage[normalem]{ulem}

\usepackage{multirow}

\usepackage{rotating}
\usepackage{arydshln}

\usepackage{marvosym}

\usepackage{stix}

\usepackage{caption}
\usepackage[margin=3cm]{geometry}
\usepackage{amscd, amsmath, appendix, authblk, color, enumerate, enumitem, float, mathtools, subcaption, tikz}
\DeclareMathOperator*{\argmin}{arg\,min}

\usepackage{amsthm}
\usetikzlibrary{matrix,arrows,positioning,automata,shapes,shapes.multipart} 

\newtheorem{theorem}{Theorem}[section]
\newtheorem{proposition}[theorem]{Proposition}
\newtheorem{lemma}[theorem]{Lemma}
\newtheorem{corollary}[theorem]{Corollary}

 \theoremstyle{definition} 
 \newtheorem{definition}[theorem]{Definition}

\theoremstyle{definition}

\newtheorem{example}[theorem]{Example}

\renewcommand{\P}{\mathcal{P}}

\title{A survey of the monotonicity and non-contradiction of consensus methods and supertree methods}

\author[1]{Mareike Fischer\thanks{\url{mareike.fischer@uni-greifswald.de, email@mareikefischer.de}}}
\author[2]{Michael Hendriksen \thanks{\url{m.hendriksen@unsw.edu.au}}}

\affil[1]{Institute of Mathematics and Computer Science, University of Greifswald, Greifswald, Germany}
\affil[2]{School of Mathematics and Statistics \protect \\University of New South Wales, Australia}

\begin{document}
\maketitle

\begin{abstract}
In a recent study, Bryant, Francis and Steel investigated the concept of \enquote{future-proofing} consensus methods in phylogenetics. That is, they investigated if such methods can be robust against the introduction of additional data like added trees or new species. 
In the present manuscript, we analyze consensus methods under a different aspect of introducing new data, namely concerning the discovery of new clades. In evolutionary biology, often formerly unresolved clades get resolved by refined reconstruction methods or new genetic data analyses. In our manuscript we investigate which properties of consensus methods can guarantee that such new insights do not disagree with previously found consensus trees, but merely refine them, a property termed \emph{monotonicity}. Along the lines of analyzing monotonicity, we also study two {established} supertree methods, namely Matrix Representation with Parsimony (MRP) and Matrix Representation with Compatibility (MRC), which have also been suggested as consensus methods in the literature. While we (just like Bryant, Francis and Steel in their recent study) unfortunately have to conclude some negative answers concerning general consensus methods, we also state some relevant and positive results concerning the majority rule ($\mathtt{MR}$) and strict consensus methods, which are amongst the most frequently used consensus methods. Moreover, we show that there exist infinitely many consensus methods which are monotonic and have some other desirable properties. 



\end{abstract}

\section{Introduction}

In phylogenetics, consensus methods play a fundamental role concerning tree reconstruction and are therefore often used in biological studies (cf.  { \cite{oreilly,carling,geisler}}). For instance, when different genes of the same set of species lead to different gene trees or when different tree reconstruction methods come to different results, it may be hard to decide which of the given trees is the \enquote{true} tree in the sense of coinciding with the underlying (unknown) species tree. This is where consensus methods come into play --- they use certain rules to summarize a set of trees to form a \emph{consensus tree}. There are various such methods used in biology, which is why it makes sense to classify them according to their properties, some of which are desirable, while others might be undesirable. Several approaches in this regard have been made in the literature. Most prominently, Day and McMorris gave an overview of certain axioms and paradigms used to categorize consensus methods \cite{day2003axiomatic}. This axiomatic approach seems to have been inspired by the famous Arrovian result on weak orders \cite{arrow1950difficulty}, then considered in the area of partitions by Mirkin \cite{mirkin1975problem}, before finally being considered in more tree-like contexts, outlined by McMorris and Powers in \cite{mcmorris2008characterization}. One popular axiom (e.g. discussed in \cite{day2003axiomatic,mcmorris2008characterization,mcmorris}) is \emph{monotonicity}, which was originally introduced in \cite{mcmorris} as an example for a desirable property of consensus methods \cite[Chapter 1.2]{day2003axiomatic} and which we will re-visit in the present manuscript, along with a new criterion which we will coin \emph{non-contradiction}.

As stated above, differing consensus methods can lead to differing consensus trees. However, even if you stick to one consensus method, its outcome might change when new input data are discovered. This is why it is important to determine how \enquote{future-proof} consensus methods are. In a recent study \cite{bryant2017can}, Bryant, Francis and Steel investigated the properties of consensus methods in an axiomatic manner (similar to the approach suggested in \cite{day2003axiomatic}), proposing three simple conditions that a consensus method should obey, referring to any such method as \emph{regular}. The purpose of the article was to investigate how \enquote{future-proof}  consensus methods are. In particular, the authors investigated associative stability (i.e., robustness against the introduction of additional trees) as well as extension stability (i.e., robustness against the introduction of additional species). Unfortunately, the authors found that such future-proofing so-called regular methods (which we will formally define later on) is impossible, as these consensus methods cannot be extension stable \cite{bryant2017can}. Moreover,  while regularity and associative stability are possible at the same time, a consensus method cannot be regular, associatively stable and Pareto (another desirable property) on so-called rooted triples \cite{delucchi2019}. In summary, the combination of certain desirable properties is not possible in any consensus method.

In the present paper we investigate a related question -- can a consensus method be robust against refinement of the input trees? This question {is relevant}, as often formerly unresolved clades in known phylogenetic trees get resolved by new genetic analyses or refined tree reconstruction methods. This implies that a given set of input trees might be changed in the sense that new clades are added rather than new species or entirely new trees as in \cite{bryant2017can}, and the main purpose of the present manuscript is to analyze the impact of this scenario on consensus methods. In the literature, consensus methods which have the property that refining some of the input trees might refine the consensus tree, but not generate a new consensus tree contradicting the previous one, are referred to as \emph{monotonic}. 

In addition to monotonicity of consensus methods, we also study two supertree methods, namely Matrix Representation with Parsimony (MRP) and Matrix Representation with Compatibility (MRC). These methods are sometimes considered as consensus methods in the literature in cases in which all input trees employ the same set of leaves. However, there is some ongoing debate about the usefulness of supertree methods in the context of consensus trees \cite{binindabook,bininda2003}, because these methods do not necessarily lead to a unique tree and thus in some cases require another consensus method to summarize all resulting trees. However, we will show in the present manuscript that even in the ideal case in which these methods do lead to a unique tree, they are not future-proof -- neither in the sense of adding new species (thus complementing the study of \cite{bryant2017can}) nor in the sense of resolving new clades.

While the above mentioned findings might be considered \enquote{bad news}, we also show that majority rule consensus methods (including the strict consensus) are indeed monotonic, i.e., they are robust against the new resolution of formerly unresolved input trees. Moreover, {we show} that there exist infinitely many consensus methods which are regular and monotonic, and even \emph{non-contradictory} (which is yet another desirable property we will define subsequently) -- even if the class of such methods does not contain some of the established methods like loose consensus or Adams consensus. This paves the way for future research, namely concerning the search for new consensus methods which have these desirable properties and are biologically plausible. Such methods could have huge potential in replacing some of the existing consensus methods.

\section{Preliminaries}

Before we can present our results, we need to formally introduce the most important concepts discussed in this manuscript. 

\subsection{Basic phylogenetic concepts}  \label{sec:basicConcepts}

In the following, let $X$ be a finite set, typically a set of \emph{taxa} or \emph{species}, but for simplicity, we may also assume without loss of generality that, whenever $|X|=n$, we have $X=\{1,\ldots,n\}$. A phylogenetic $X$-tree $T$ is a connected acyclic graph whose leaves are bijectively labelled by the elements of $X$. If there is one distinguished node $\rho$ referred to as the \emph{root}  of the tree, then $T$ is called \emph{rooted}; otherwise $T$ is called \emph{unrooted}. Trees are often denoted in the (non-unique) \emph{Newick format} (cf. \cite{newick}), which uses nested brackets in such a way that two closely related species are grouped closely together, cf. Figure \ref{fig:all4taxontrees}. Typically, the uppermost level of the nesting is supposed to refer to the root if $T$ is rooted. Possible Newick formats for several unrooted trees can be found in the caption of  Figure \ref{fig:all4taxontrees}. The rooted trees from Figure \ref{AdamsAho} can be denoted in the Newick format as follows: $(((1,2),3),(4,(5,6)))$ denotes $T_1$, $(((1,5),4),(3,(2,6)))$ denotes $T_2$, and $((2,3),1,6,(4,5))$ denotes $T_3$.

Let $RP(X)$ and $UP(X)$ denote the set of rooted and unrooted phylogenetic trees on $X$, respectively. Then, a \emph{profile} of trees is an ordered tuple $(T_1,...,T_k)$ of trees such that $T_1,...,T_k \in RP(X)$ or $T_1,...,T_k \in UP(X)$ (that is, the trees in a profile must all be rooted or all be unrooted, and they must all refer to the same taxon set $X$). 

We now first turn our attention to unrooted trees. A bipartition $\sigma$ of $X$ into two non-empty and disjoint subsets $A$ and $B$ is called an \emph{$X$-split} (or \emph{split} for short), and is denoted by $\sigma=A|B$. There is a natural relationship between $X$-splits and the edges of an unrooted phylogenetic $X$-tree $T$, because the removal of an edge $e$ induces such a bipartition of $X$.  
An $X$-split  {with $\min\{|A|,|B|\}=1$} is called {\em trivial}. In the following, the set of all induced (non-trivial) $X$-splits of $T$ will be denoted by $\Sigma(T)$ (or $\Sigma^*(T)$, respectively). Note that two $X$-splits $\sigma_1=A|B$ and $\sigma_2=\widetilde{A}|\widetilde{B}$ are called \emph{compatible} if and only if at least one of the intersections $A \cap \widetilde{A}$, $A \cap \widetilde{B}$, $B \cap \widetilde{A}$ or $B \cap \widetilde{B}$ is empty. It is well-known (the so-called \emph{Buneman theorem} or \emph{splits-equivalence theorem}) that each set of pairwise compatible splits uniquely defines a tree, which can be found by the  \texttt{Tree Popping} algorithm \cite{meacham,Semple2003}.

Non-trivial splits can also be coded as so-called \emph{binary characters} {as follows: If $\sigma=A|B$, the corresponding binary character is simply the function $f:X\rightarrow \{0,1\}$ with 

\[ f(x)= \begin{cases}1 & \mbox{if }x\in A \\ 0 & \mbox{if } x \in B \end{cases}, \]

and this often gets denoted by the shorthand $f(1)f(2)\ldots f(n)$. Moreover, we may assume without loss of generality that $f(1)=1$, because we consider characters equivalent if they correspond to the same $X$-split, even if the roles of 0 and 1 are interchanged. We denote the set of non-trivial binary characters induced by a phylogenetic $X$-tree with $B^*(T)$.}

\begin{figure}
	\centering
	\includegraphics[scale=.25]{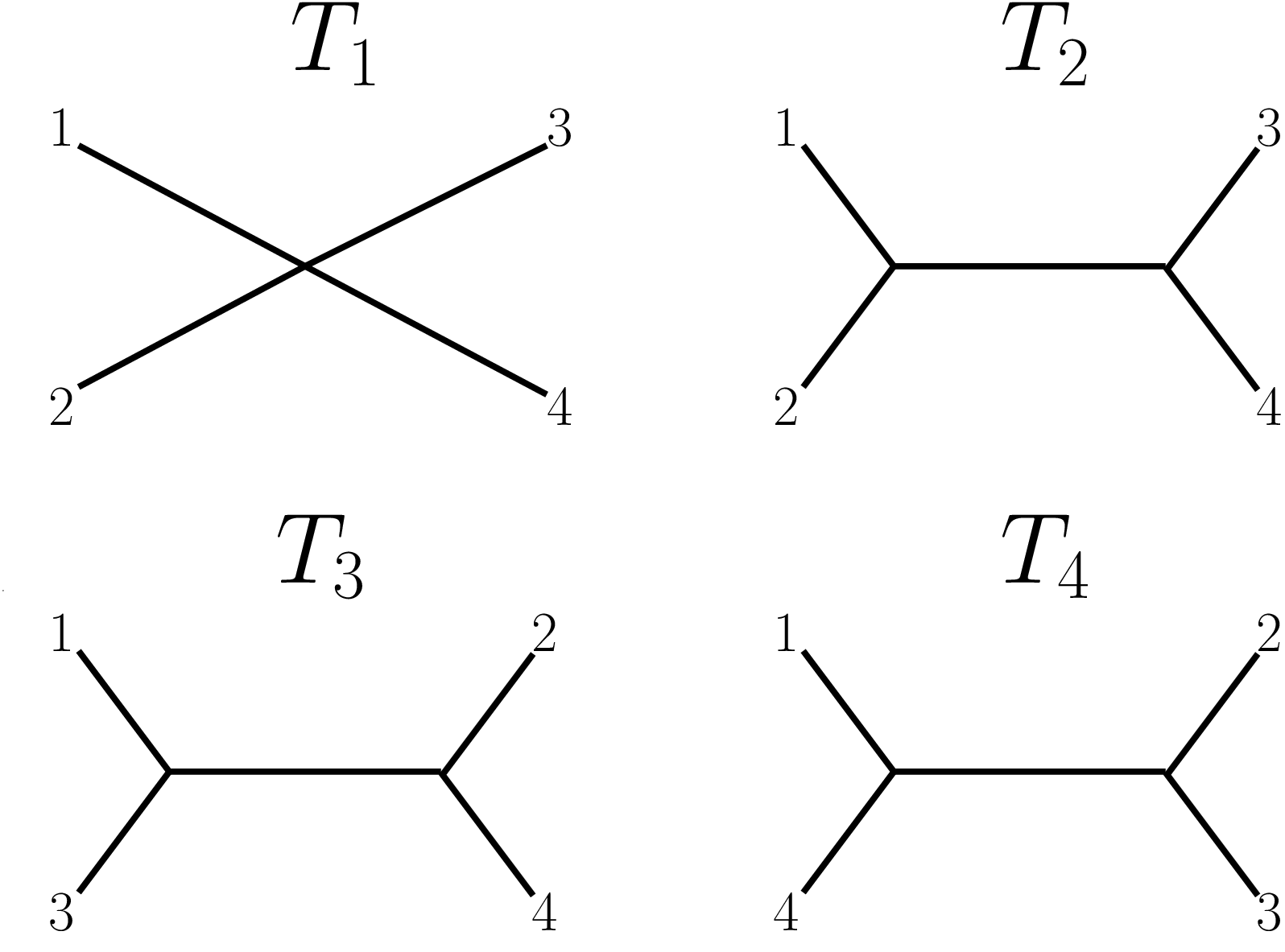}
	\caption{All phylogenetic $X$-trees for $X=\{1,2,3,4\}$. The Newick formats of the given trees are $T_1=(1,2,3,4)$, $T_2=((1,2),3,4)$, $T_3=((1,3),2,4)$ and $T_4=((1,4),2,3)$. Here, we have $\Sigma^*(T_1)=B^*(T_1)=\emptyset$ and  $\Sigma^*(T_2)=\{12 | 34\}$, and thus $B^*(T_2)=\{1100\}$, for instance. Moreover, we have $T_1\preceq T_2$, $T_1\preceq T_3$ and $T_1\preceq T_4$, but there is no such relation between $T_2$, $T_3$ and $T_4$.	}
	\label{fig:all4taxontrees}
\end{figure}

In the rooted setting, instead of considering splits  we consider  \emph{clusters}. A clade of a rooted phylogenetic tree $T$ is a pendant subtree $T'$ of  $T$, and a cluster is the set of leaves $Y \subseteq X$ of such a pendant subtree $T'$ of $T$. Note that every vertex of a rooted phylogenetic tree $T$ induces a cluster as it can be regarded as the root of such a subtree of $T$. Clusters of sizes $1$ and $n$ are called \emph{trivial}. We denote the set of all clusters of a rooted phylogenetic tree $T$ by $\mathcal{C}(T)$, and the set of its non-trivial clusters by $\mathcal{C}^*(T)$. Note that every cluster $Y\subsetneq X$  induces an $X$-split, namely $\sigma=Y | (X \setminus Y)$ (note that if $Y=X$, $Y$ does not induce an $X$-split, which is why the condition $Y\subsetneq X$ is crucial). In the following, we call two clusters $c_1, c_2$ \emph{compatible} if we have $c_1 \subseteq c_2$ or $c_2 \subseteq c_1$ or $c_1 \cap c_2=\emptyset$. Note that this necessarily implies that the splits induced by $c_1$ and $c_2$ are also compatible if and only if $c_1$ and $c_2$ are compatible.

{It is a direct consequence of the Buneman theorem that a rooted phylogenetic tree $T$ is fully determined by $\mathcal{C}^*(T)$, and $T$ can be reconstructed from  $\mathcal{C}^*(T)$ with the \texttt{BUILD} algorithm \cite{aho}. } 
The last concept we need to recall before we can turn our attention to consensus methods is the partial order $\preceq$. For two rooted (or unrooted) phylogenetic trees $T_1$ and $T_2$, we say $T_2$ \emph{refines} $T_1$ or $T_2$ \emph{is a refinement of} $T_1$ and denote this by $T_1 \preceq T_2$, whenever $\mathcal{C}(T_1) \subseteq \mathcal{C}(T_2)$ (or $\Sigma(T_1) \subseteq \Sigma(T_2)$, respectively, cf. Figure \ref{fig:all4taxontrees}). Note that a tree that has no refinement except for itself is called \emph{fully resolved}.

\subsection{Consensus methods}

We are now finally in the position to introduce the most important concept underlying this manuscript. 

\begin{definition}
A \emph{rooted consensus method} (resp.  \emph{unrooted consensus method}) is a function $\phi$ that, for every set $X$ of taxa and every number $k \ge 1$, associates with each profile of $k$ trees from $RP(X)$ (resp. $UP(X)$) a unique corresponding tree in $RP(X)$ (resp. $UP(X)$).
\end{definition}

Note that the trees occuring in a profile $\mathcal{P}$ of trees will subsequently often be referred to as \emph{input trees}, and the tree returned by a consensus method will then often be called the \emph{output tree}.

Following \cite{bryant2017can}, we call a consensus method $\phi$ \emph{regular} if it obeys the following three properties:

\begin{enumerate}
\item \textbf{Unanimity:} If the trees in $\mathcal{P}$ are all the same tree $T$, then $\phi(\mathcal{P})=T$.
\item \textbf{Anonymity:} Changing the order of the trees in $\mathcal{P}$ does not change $\phi(\mathcal{P})$.
\item \textbf{Neutrality:} Changing the labels on the leaves of the trees in $\mathcal{P}$ simply relabels the leaves of $\phi(\mathcal{P})$ in the same way.
\end{enumerate}

{The stated properties are needed to make the consensus biologically plausible and mathematically interesting. More detailed background concerning regular consensus methods can be found in \cite{bryant2017can}. }

A reasonable request for a consensus tree might be that every cluster in the output has some evidence in the input tree - e.g., is present in at least one tree in the input profile. Note that some consensus methods $\phi$ have this property, i.e., if $\mathcal{P}=(T_1,\ldots,T_k)$ is a profile of phylogenetic $X$-trees, for every cluster $c$ (or split $\sigma$, respectively) in $\phi(\mathcal{P})$, there is at least one tree $T_i$ that displays $c$ (or $\sigma$). Consensus methods with this property are referred to as \emph{co-Pareto} in the literature \cite{wilkinsoncotton,husondezulian,mcmorris}.

However, one of the primary aims of this paper is to consider and analyze the following two concepts, namely monotonicity, which has already been discussed in the literature, and non-contradiction, which is a property we will introduce in the present manuscript. We start with monotonicity.

\begin{definition}
Let $\phi$ be a consensus method. If, for any pair of profiles $\P=(T_1,\dots,T_k),\P'=(T_1',\dots,T_k')$ with $T_i \preceq T_i'$ for each $i \in \{1,...,k\}$, we have $\phi(\P) \preceq \phi(\P')$, then $\phi$ is referred to as \emph{monotonic}.
\end{definition}

Another natural property that might be desirable for a consensus method is that of being \emph{non-contradictory}. This can be viewed as a relaxation of the co-Pareto property: One might wish for a consensus tree which, at least, does not contradict \emph{every} tree in the input profile (even if the input profile does not contain a single tree that displays a certain clade or split). We now formally define this property.

\begin{definition}
    A consensus method is termed \emph{non-contradictory} if each cluster (or split, respectively) of the output tree is compatible with at least one tree in the input profile, where a cluster (or split) is called compatible with a tree if it is pairwise compatible with all its induced clusters (or splits, respectively).
\end{definition}

We now list three of the most frequently used consensus methods. 

\begin{definition} \label{def:consensus} Let $\mathcal{P}$ be a profile of trees $(T_1,\ldots,T_k)$. {Then the following rules define ways to form a consensus tree:}

\begin{itemize}
\item {\bf Strict consensus ($\mathtt{strict}$): } 
The strict consensus tree $\mathtt{strict}(\mathcal{P})$ contains precisely \emph{all} clusters (in the rooted setting), respectively all splits (in the unrooted setting) that are present in all trees $T_i \in \mathcal{P}$.
\item {\bf Loose (or semi-strict) consensus ($\mathtt{loose}$): } 
The loose consensus tree $\mathtt{loose}(\mathcal{P})$ contains precisely all clusters (or splits, respectively) that are present in at least one tree $T_i \in \mathcal{P}$ and that are not incompatible with any split induced by any tree in $\mathcal{P}$.
\item {\bf Majority rule consensus ($\mathtt{MR}$): } Let $p> 50$. Then, 
the majority-rule consensus tree 
$\mathtt{MR}_p(\mathcal{P})$ contains precisely all clusters (or splits, respectively) that are present in at least $p\%$ of the trees $T_i \in \mathcal{P}$. Note that whenever there  is no ambiguity concerning $p$ or whenever a statement holds for all possible choices of $p>50$, we may simply write $\mathtt{MR}(\mathcal{P})$ rather than $\mathtt{MR}_p(\mathcal{P})$.
\end{itemize}
\end{definition}

{For all three consensus methods mentioned in Definition \ref{def:consensus}, the respective consensus tree $\phi(\mathcal{P})$ exists and is unique. For majority rule (and thus also strict) consensus, this was shown in \cite{margush1981consensus}, and for loose consensus this follows directly from the definition.}

\begin{remark}\label{rem:gammaMRGamma} Let $\P$ be a profile of (rooted or unrooted) phylogenetic $X$-trees. Then, it can be easily seen by Definition \ref{def:consensus} that $\mathtt{strict}(\P) \preceq \mathtt{MR}(\P)$. 
\end{remark}

Next, we want to introduce two more classic consensus methods, namely the Adams and the Aho consensus, which are both only defined for rooted trees. Our definitions are based on \citep{bryant2017can}, but for further details, we refer the reader also to \citep{bryant}. 

We start with \textbf{Adams consensus}. In order to build the Adams consensus tree for a profile $\mathcal{P}$ of rooted phylogenetic $X$-trees, we start by considering the partition $\Pi(X)$, which is defined as the non-empty intersections of the maximal clusters of the trees in $\mathcal{P}$ other than $X$. Note that $\Pi(X)$ will correspond to the maximal clusters in the consensus tree. Once this partition $\Pi(X)$ is determined, we take each element of $\Pi(X)$ and recursively repeat this procedure for the respective subset of taxa until it has size one and thus cannot be refined anymore. This eventually produces a set of compatible clusters on $X$, which then can be assembled into a unique tree. An example for the Adams consensus tree is given in Figure \ref{AdamsAho}. 

Similarly, for \textbf{Aho consensus}, we construct a partition $\Pi(X)$ that also gets recursively refined, but in this case, $\Pi(X)$ equals the connected components of the graph $(X,E_{\mathcal{P}})$, where there is an edge $\{x_1,x_2\}  \in  E_{\mathcal{P}}$ precisely if there  exists $x_3 \in X \setminus \{x_1,x_2\}$ such that the subtree $((x_1,x_2),x_3)$ is contained in all trees of $\mathcal{P}$. An example for the Aho consensus tree is given in Figure \ref{AdamsAho}.

In both cases, the hierarchy induced by the recursive partitioning can be used to reconstruct a unique rooted tree (which we refer to as $\varphi_{Ad}(\mathcal{P})$ or $\varphi_{Aho}(\mathcal{P})$, respectively)  using the so-called BUILD algorithm \citep{aho}.

\begin{remark} By \cite[p. 613]{bryant2017can}, \enquote{All standard phylogenetic consensus methods (e.g., strict consensus, majority rule, loose consensus, and Adams consensus)}, are regular. It  is easy to see that this, too, applies to Aho consensus.
\end{remark}

We now turn our attention to two supertree methods, which also sometimes appear in the context of consensus trees.

\subsection{MRP and MRC} \label{sec:MRPMRCintro}

In the literature {one} can find various methods to construct supertrees from (multi)sets of input trees. One difference to our setting is that these input trees need not coincide in their taxon sets. However, in the supertree setting, too, input trees might come with conflicting information, in which case the supertree corresponds to some sort of consensus. This is why supertree methods have also been regarded as consensus methods by some authors  \citep{bryant, bininda2003, levasseur}, even if there is an ongoing debate as to whether this is justified \citep{binindabook,bininda2003}. 

The reason why it is not straightforward to regard supertree methods as consensus methods is that a supertree need not be unique. This problem can be overcome, e.g., by using an existing consensus method in order to summarize all supertrees into a single consensus tree \cite{bryant}.  Another argument against regarding supertree methods as consensus methods is that a supertree may contain splits (or clusters) that are not present in any of the input trees (even though this can also happen with some established consensus methods, see Figure \ref{AdamsAho}: There, cluster $\{4,5\}$ is contained in the Adams and Aho consensus trees but in none of the input trees). 
 However, in the present manuscript, we will not join this debate, but we will consider two {established} supertree methods as consensus methods in the case in which the respective supertree \emph{is} unique and test them for refinement stability: \emph{Matrix Representation with Parsimony}, or \emph{MRP} for short, and \emph{Matrix Representation with Compatibility}, or \emph{MRC} for short.\footnote{Note that -- as we only consider cases where the output of MRP and MRC is unique -- applying a regular consensus method like majority-rule to the output of MRP or MRC as suggested for instance in \cite{bryant} would not change anything due to unanimity. } Note that as opposed to the consensus methods introduced above, MRP and MRC both are only defined for unrooted trees as they deliver no information about the root position.

In order to understand MRP and MRC, we first need to introduce the concept of \emph{parsimony}. In this regard, an \emph{extension} of a binary character $f:X\rightarrow \{0,1\}$ on a phylogenetic $X$-tree $T$ with vertex set $V(T)$ is a function $g^f: V(T) \longrightarrow \mathcal \{0,1\}$, such that $g^f_{|_X}=f$, i.e., $g^f$ agrees with $f$ on the leaves of $T$ but also assigns states to inner vertices of $T$.
Now, the \emph{changing number} $ch(g^f,T)$ of an extension $g^f$ on a phylogenetic tree $T$ is simply the number of edges $e=\{u,v\}$ for which $g^f(u)\neq g^f(v)$. An extension $\widetilde{g^f}$ such that $ \widetilde{g^f} = \argmin\limits_{g^f} ch(g^f,T)$ (where the minimum is taken over all extensions of $f$ on $T$) is called \emph{minimal}, and the changing number of such a minimal extension is called \emph{parsimony score} of $f$ on $T$, denoted $ps(f,T)$. Thus, we have $ps(f,T)= \min\limits_{g^f} ch(g^f,T)$. The parsimony score of a character $f$ on a phylogenetic tree $T$ can for instance be calculated by the {well-known} \emph{Fitch-Hartigan algorithm} \cite{Fitch,Hartigan1973}. Moreover, the parsimony score of a (multi)set of characters, which is also often referred to as \emph{alignment} in evolutionary biology, is simply defined as the sum of the parsimony scores of all characters.  

Next, the \emph{maximum parsimony tree}, or \emph{MP tree} for short, for a (multi)set $S$ of characters on $X$ is a phylogenetic $X$-tree $T$ for which we have $ch(S,T)=ps(S,T)$, i.e., a tree which has minimal changing number for $S$ amongst all phylogenetic $X$-trees. Note that this tree need not be unique. 

Given a profile $\mathcal{P}=(T_1,\ldots,T_m)$ of phylogenetic $X$-trees, we consider the union of binary characters $\mathcal{B}(\mathcal{P}):=\bigcup\limits_{i=1}^m B^*(T_i)$. Note that biologists often regard the multiset $\mathcal{B}(\mathcal{P})$ as a so-called alignment, i.e., as an ordered sequence of characters. However, as the order of characters does not play a role for the methods considered in the present manuscript, it is sufficient for us to consider the union. However, note that here we understand the union in the sense of a multiset, i.e., as two phylogenetic trees $T_i$ and $T_j$ might have a split in common, they might contribute the same character to $\mathcal{B}$, in which case both copies are kept.
 Now, an MP tree of  $\mathcal{B}$ is called \emph{Matrix Representation with Parsimony tree} of $\mathcal{P}$, or \emph{MRP tree} of $\mathcal{P}$ for short. Again, the MRP tree need not be unique. In particular, it can easily be seen that if $T$ is an MRP tree for some profile $\P$, all of its refinements are MRP trees for $\P$, too (because adding additional edges to a tree cannot increase its parsimony score).

The second supertree method which we want to consider is \emph{Matrix Representation with Compatibility}, or \emph{MRC} for short. This method analyzes all characters of $\mathcal{B}(\P)$ and finds a maximum compatible subset. Here, a set $S$ of binary characters is called \emph{compatible} if all characters in $S$ are pairwise compatible, and two binary characters are called compatible if their corresponding splits are compatible. By the Buneman theorem, there is a unique tree $T$ corresponding to each such set of compatible characters, which is why we say that such a set \emph{induces} $T$. Now, MRC works as follows: After a maximum compatible subset of characters in $B(\mathcal{P})$ is found and the corresponding tree is constructed, this tree is called \emph{Matrix Representation with Compatibility tree} or \emph{MRC tree} for short. Note that there may be more than one MRC tree as more than one maximum set of compatible characters may exist in $\mathcal{B}(\P)$.

\section{Results}

\subsection{Established consensus methods}
In this section, we turn our attention to some established consensus methods and analyze them with regards to monotonicty and non-contradiction. 

\subsubsection{Monotonicity}

First, we turn our attention the question of which established consensus methods are actually monotonic and which ones are not. We start with a positive result, concerning strict and majority-rule consensus. Note that whenever not stated otherwise, the results hold both in the rooted as well as in the unrooted setting. 

For completeness, we cite the following result from the literature \cite{mcmorris2008characterization,mcmorris}. 
\begin{theorem}[Theorem 1 in \cite{mcmorris2008characterization}, Theorem 2 in \cite{mcmorris}]\label{thm_strictAndMLareRS} Strict consensus and majority-rule consensus are monotonic.
\end{theorem}

Now we turn our attention to the loose consensus method and state our first negative result. 

\begin{proposition} \label{prop_adamsaholooseNOTmon}
Loose consensus, Adams consensus and Aho consensus are \emph{not} monotonic. 
\end{proposition}

\begin{proof} We prove the assertion concerning $\mathtt{loose}$ by providing an explicit counterexample. Let $X=\{1,2,3,4\}$ and $T_1=((1,2),3,4)$ be either rooted or unrooted, and let $S=(1,2,3,4)$ be the (rooted or unrooted) star-tree, i.e., the unique tree on $X$ with only one inner vertex. Let $\mathcal{P}=(T_1,S)$. Then, for the loose consensus tree we have $\mathtt{loose}(\mathcal{P})=T_1$. However, if we refine $S$ to become $T_2=((2,3),1,4)$ and consider the profile $\mathcal{P}'=(T_1,T_2)$, we get $\mathtt{loose}(\mathcal{P}')=S$ (because the splits $12|34$ and $23|14$ are incompatible). So if we refine input tree $S$, we get a coarser consensus tree, namely $S$ instead of $T_1$, not a refined one. This completes the proof for $\mathtt{loose}$.

 Concerning Adams and Aho consensus, we prove the statement by presenting an explicit counterexample. Let $T_1=(((1,2,3),4),5)$, $T_2=((1,2),(3,4),5)$ and $T_2'=(((1,2),(3,4)),5)$, and let 
$\mathcal{P}=(T_1,T_2)$ and $\mathcal{P}'=(T_1,T_2')$. Note that $T_2'$ is a refinement of $T_2$, and hence any consensus method $\phi$ that is monotonic must have $\phi(\P) \preceq \phi(\P')$. However, the Adams consensus tree of $\mathcal{P}$ is $\varphi_{Ad}(\mathcal{P})=T_2$, and the Adams consensus of $\mathcal{P}'$ is tree $\varphi_{Ad}(\mathcal{P})'=(((1,2),3,4),5)$. It follows that the Adams consensus method is \emph{not} monotonic as $\varphi_{Ad}(\mathcal{P}) \not\preceq \varphi_{Ad}(\mathcal{P}')$.

Note that in this  example, Aho consensus coincides with the Adams consensus tree, so Aho consensus method is also \emph{not} monotonic, which completes the proof. 

\end{proof}

While Proposition \ref{prop_adamsaholooseNOTmon} can be seen as a negative result, there are cases when monotonicity can be guaranteed even for the loose, Adams and Aho consensus methods. This depends on certain properties of the input sets. We now analyze this further and start with the loose consensus method.

\begin{proposition} \label{prop_loose_linearordering} Let $\P=\{T_1,\ldots,T_k\}$ and let $\P'=\{T_1',\ldots,T_k'\}$ be such that $T_i\preceq T_i'$ for all $i\in \{1,\ldots,k\}$. Moreover, let $\P'$ be such that it permits a linear ordering, i.e., for all $i,j\in \{1,\ldots,k\}$ we either have $T_i'\preceq T_j'$ or $T_j'\preceq T_i'$. Let $\phi \in \{\mathtt{loose},\varphi_{Adams},\varphi_{Aho}\}$. Then we have: $\phi(\P) \preceq \phi(\P')$.  
\end{proposition}

\begin{proof} Let $\P$ and $\P'$ be as described in the proposition. Then, as $\P'$ allows for a linear ordering, we may assume without loss of generality that $T_i' \preceq T_{i+1}'$ for all $i=1,\ldots,k-1$. Note that this implies $T_k'$ is the most refined tree in $\P'$, and as all splits of $T_k'$ are necessarily compatible with all other trees in $\P'$, we must have $\mathtt{loose}(\P')=T_k'$. Moreover, the Adams consensus as well as the Aho consensus for two trees $\widetilde{T}$ and $\widehat{T}$ with $\widetilde{T}\preceq \widehat{T}$ will always return $\widetilde{T}$, we have $\varphi_{Adams}(\P')=\varphi_{Aho}(\P')=T_1'$, i.e., the coarsest tree in $\P'$ is  returned by these methods.

On the other hand, as we have $T_i\preceq T_i'$ for all $i=1,\ldots,k$ by assumption, we know that the set $\Sigma(\P)$ containing all splits occurring in any of the trees of $\P$ is a subset of $\Sigma(\P')$, the set of splits of $\P'$, i.e., $\Sigma(\P)\subseteq \Sigma(\P')$. By $T_i'\preceq T_k'$ for all $i=1,\ldots,k$, we know that $\Sigma(\P')=\Sigma(T'_k)=\Sigma(\mathtt{loose}(\P'))$. As by definition we must have $\Sigma(\mathtt{loose}(\P)) \subseteq \Sigma(\P)$, we conclude $\Sigma(\mathtt{loose}(\P)) \subseteq \Sigma(\P)\subseteq \Sigma(\P')=\Sigma(\mathtt{loose}(\P'))$, which shows that $\mathtt{loose}(\P)\preceq \mathtt{loose}(\P')$ as desired.  Moreover, again as above, $\varphi_{Adams}(\P)=\varphi_{Aho}(\P)=T^\star$, where $T^\star$ is the coarsest tree in $\P$. Let $i^\star\in \{1,\ldots,k\}$ be such that $T_{i^\star}=T^\star$. Then, we have: $\varphi_{Adams}(\P)=\varphi_{Aho}(\P)=T^\star=T_{i^\star} \preceq T_1 \preceq T_1' =\varphi_{Adams}(\P')=\varphi_{Aho}(\P')$, which completes the proof.
\end{proof}

So in summary, of the established consensus methods, strict and majority-rule are monotonic, whereas the loose, Adams and Aho consensus methods are not, but the latter at least have the monotonicity property if the input profiles are such that the refined one allows for a linear ordering. In Section \ref{sec:MRP} we will show that these monotonicity statements also apply to MRP and MRC. But before we do so, we turn our attention to another property that consensus methods might have, namely that of being non-contradictory.

\subsubsection{Non-contradiction}

Non-contradiction is another property that might be desirable for a biologically meaningful consensus method. It states that each cluster (or split) in the consensus tree must be compatible with at least one input tree. Non-contradiction can be regarded as a weaker version of being co-Pareto, i.e., of the requirement for certain clusters to be present in at least one input tree, as, e.g., required by loose consensus. While it might not always be necessary that all clusters found by a consensus tree need to be present in at least one input tree (as this might make the resolution of the consensus tree too coarse), the consensus tree is often required to at least not contain any clusters that contradict \emph{all} input trees, which is guaranteed if a consensus method is non-contradictory.

The main aim of this short subsection is to state and prove that this is a property that indeed various established consensus methods share, but not all of them.

\begin{proposition}\label{prop:MRnoncontra} 
Loose consensus, strict consensus and majority rule consensus are all non-contradictory, but Adams and Aho consensus are not.
\end{proposition}

\begin{proof} By Definition \ref{def:consensus}, loose consensus, strict consensus and majority rule consensus all lead to trees that only contain clusters (or splits) that are present in at least one tree. So each cluster (or split) in the respective consensus tree must be compatible with at least one tree in the input profile, namely with the one it is induced by. This completes the proof of the first statement. 

Concerning Adams and Aho consensus, consider trees $T_1$ and $T_2$ in Figure \ref{AdamsAho}. For $\mathcal{P}=(T_1,T_2)$, Adams and Aho consensus coincide, and $\varphi_{Ad}(\mathcal{P})=\varphi_{Aho}(\mathcal{P})=T_3$, but $T_3$ contains the clusters $\{2,3\}$ and $\{4,5\}$, both of which are incompatible with both $T_1$ and $T_2$. Hence neither Adams nor Aho consensus are non-contradictory, which completes the proof.
\end{proof}

\begin{figure}
    \centering
\includegraphics[scale=.3]{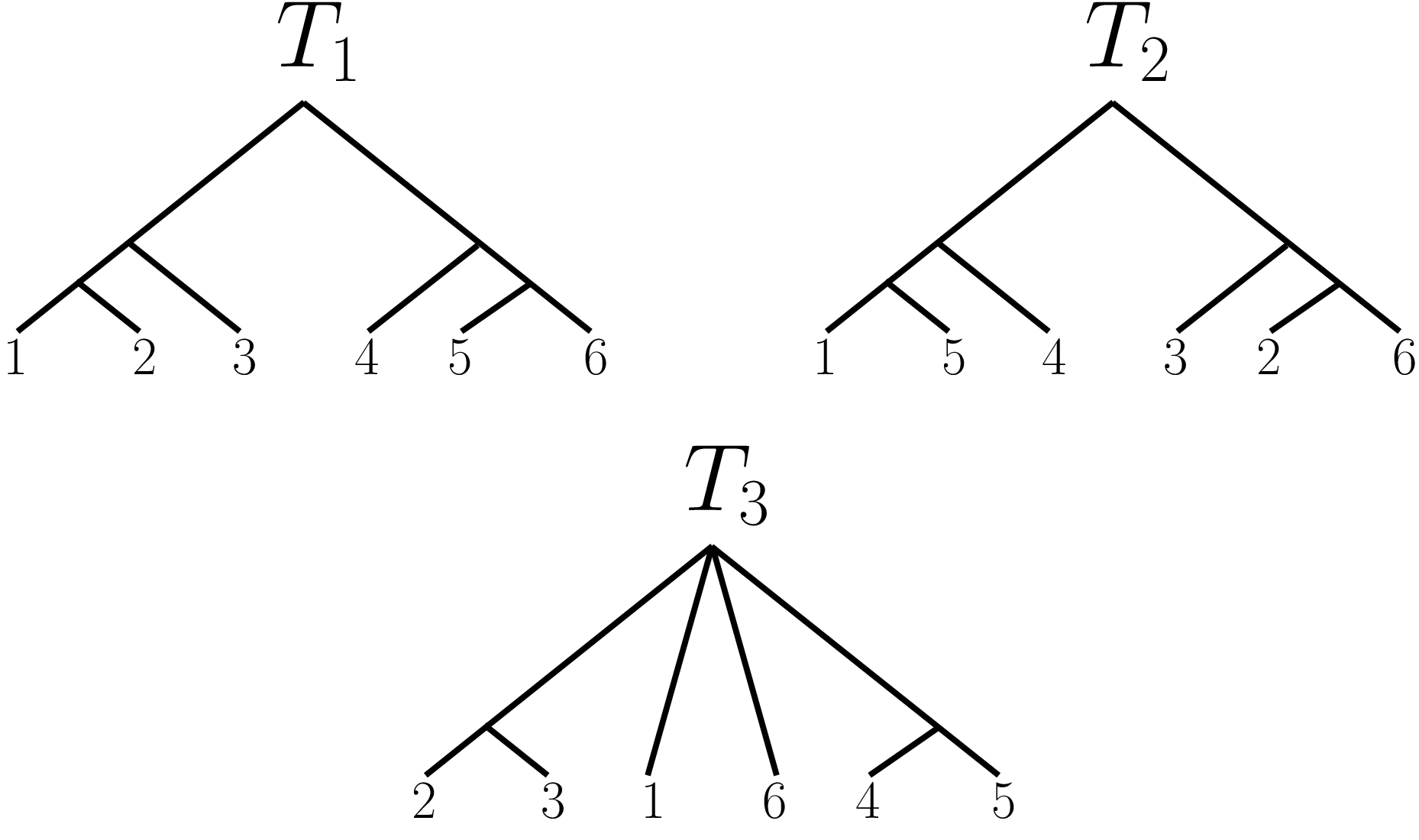}
	\caption{The Adams and Aho consensus methods are not non-contradictory as $\varphi_{Ad}(T_1,T_2)=\varphi_{Aho}(T_1,T_2)=T_3$.} 
	\label{AdamsAho}
	\end{figure}

We will now turn our attention to general the supertree methods MRP and MRC for the cases when they can easily be used as consensus methods, namely when their output is a unique tree.

\subsection{Results on the established supertree methods MRP and MRC used as consensus methods} \label{sec:MRP}

In this section, we take a closer look at MRP and MRC as consensus methods. In particular, we analyze them concerning monotonicity and non-contradiction. 

We start with the first main theorem of this section. 

\begin{theorem}\label{thm:mrpmrc}
MRP and MRC are \emph{not} monotonic.
\end{theorem}

\begin{proof}
Consider again Figure \ref{fig:all4taxontrees}. We consider profiles $\mathcal{P}_1=(T_1,T_1,T_2)$ and $\mathcal{P}_2=(T_3,T_3,T_2)$. We consider the non-trivial splits induced by these profiles and code them as binary characters as explained in Section \ref{sec:basicConcepts}. Then we can define $\mathcal{B}(\mathcal{P}_1)$ (and, analogously, $\mathcal{B}(\mathcal{P}_2)$)  as in Section \ref{sec:MRPMRCintro}. Then, for $\mathcal{P}_1$ we have $\mathcal{B}(\mathcal{P}_1)=\emptyset \cup \emptyset\cup \{1100\}=\{1100\}$, as the star tree $T_1$ has no non-trivial splits. However, if we refine both star trees such that we derive $\mathcal{P}_2$ (note that $T_3$ is a refinement of $T_1$), we derive $\mathcal{B}(\mathcal{P}_2)=\{1010,1010,1100\}$. 

We now start with considering MRP: 
It can be easily verified that  the unique MRP tree for $\P_1$ is $T_2$: We have $ps(1100,T_1)=ps(1100,T_3)=ps(1100,T_4)=2$, but $ps(1100,T_2)=1$, so in total $ps(\mathcal{B}(\mathcal{P}_1),T_2)<ps(\mathcal{B}(\mathcal{P}_1),T_i)$ for $i=1,3,4$. However, as $ps(1010,T_3)=1$ and $ps(1010,T_1)=ps(1010,T_2)=ps(1010,T_4)=2$, we also have that $ps(\mathcal{B}(\mathcal{P}_2),T_1)=ps(\mathcal{B}(\mathcal{P}_2),T_4)=2+2+2=6$, $ps(\mathcal{B}(\mathcal{P}_2),T_2)=1+2+2=5$ and  $ps(\mathcal{B}(\mathcal{P}_2),T_3)=2+1+1=4$. So in total, $T_3$ is the unique MRP tree for $\mathcal{P}_2$. This shows that the MRP tree of $\mathcal{P}_1$ does not get refined by that of $\mathcal{P}_2$, even though $\mathcal{P}_2$ is a refinement of $\mathcal{P}_1$. Thus, MRP is \emph{not} monotonic. 

Now let us consider MRC:
In $\mathcal{B}(\mathcal{P}_1)$, there is only one character, namely 1100, and this character induces tree $T_2$. Thus, $T_2$ is the unique MRC tree for $\mathcal{P}_1$. 
On the other hand, in $\mathcal{B}(\mathcal{P}_2)=\{1100,1010,1010\}$ we have both 1100 and 1010, which are incompatible (as the splits $12|34$ and $13|24$ are incompatible). Thus, not all three characters of $\mathcal{B}(\mathcal{P}_2)$ can be compatible. The maximum number of compatible characters in $\mathcal{B}(\mathcal{P}_2)$ is therefore 2. As no subset containing both 1100 and 1010 can be compatible, it turns out that the unique maximal subset of compatible characters is $\{1010,1010\}$, and this subset induces tree $T_3$. Therefore, $T_3$ is the unique MRC tree for $\mathcal{P}_2$. This shows that the  MRC tree for $\mathcal{P}_1$ does not get refined by that of $\mathcal{P}_2$, even though $\mathcal{P}_2$ is a refinement of $\mathcal{P}$. Thus, MRC is \emph{not}  monotonic. 

This completes the proof.
\end{proof}

Similarly to the previous section, where we showed that while monotonicity cannot be generally guaranteed for $\mathtt{loose}$, $\varphi_{Adams}$ and $\varphi_{Aho}$, it does hold in the case where the finer input profile allows for a linear ordering (cf. Proposition \ref{prop_loose_linearordering}, we want to consider this setting now for MRP and MRC.

\begin{proposition} \label{prop_MRPMRC_linearordering} Let $\P=\{T_1,\ldots,T_k\}$ and let $\P'=\{T_1',\ldots,T_k'\}$ be such that $T_i\preceq T_i'$ for all $i\in \{1,\ldots,k\}$. Moreover, let $\P'$ be such that it permits a linear ordering, i.e., for all $i,j\in \{1,\ldots,k\}$ we either have $T_i'\preceq T_j'$ or $T_j'\preceq T_i'$. 
Then we have: $MRC(\P)$ and  $MRC(\P')$ are unique trees, and we have  $MRC(\P)\preceq MRC(\P')$. Moreover, $MRC(\P')$ is then also an MRP tree for $\P$ and $\P'$, and  $MRC(\P)$ is also an MRP tree for $\P$. If additionally we have that $\P'$ contains a binary tree $T$, then $MRP(\P)=MRC(\P')=T$. 
\end{proposition}

\begin{proof}  Let $\P$ and $\P'$ be as described in the proposition. Then, as $\P'$ allows for a linear ordering, we may assume without loss of generality that $T_i' \preceq T_{i+1}'$ for all $i=1,\ldots,k-1$. Note that this implies $T_k'$ is the most refined tree in $\P'$, and as the split set $\Sigma(\P')$ of all splits in $\P'$ contains only elements which are pairwise compatible (by $T_i' \preceq T_{i+1}'$), we know by the Buneman theorem that there is a unique tree containing all splits from  $\Sigma(\P')$. Necessarily, this tree must be $T_k'$, which shows that $MRC(\P')=T_k'$, so $MRC(\P')$ is indeed a unique tree. Moreover, as we have $T_i\preceq T_i'\preceq T_k'$ for all $i=1,\ldots,k$, we know that $\Sigma(\P)\subseteq \Sigma(T_k')$, which implies that all splits in $\Sigma(\P)$ are compatible. Thus, again by the Buneman theorem, there is a unique tree $T$ with split set $\Sigma(\P)$, which necessarily is the unique MRC tree of $\P$. Moreover, by $\Sigma(T)=\Sigma(\P)\subseteq \Sigma(\P')=\Sigma(MRC(\P'))=\Sigma(T_k')$ we conclude $MRC(\P)\preceq MRC(\P')$ as desired.

On the other hand, we must have $ps(\mathcal{B}(\P'),T_k') \leq ps(\mathcal{B}(\P'),T_i')$ for all $i:T_i'\precneq T_k'$. Note that no other tree $\widetilde{T}$ can have $ps(\mathcal{B}(\P'),\widetilde{T}) <ps(\mathcal{B}(\P'),T_k')$, as $T_k'$ contains all splits of $\Sigma(\P')$. Thus, $T_k'\in MRP(\P')$. Moreover, as we know $T_i'\preceq T_k'$ for all $i=1,\ldots,k$, if $T_k'$ is binary, we know that $MRP(\P')=T_k'=MRC(\P')$ (as in this case, $MRC(\P')$ cannot be refined any further to generate more MRP trees). It remains to show that $T_k'=MRC(\P')$ is also an MRP tree for $\P$. In order to see this, note that as we have already seen that $MRC(\P)$ is a unique tree containing \emph{all} splits from $\Sigma(\P)\subseteq \Sigma(\P')=\Sigma(T_k')$, it necessarily is also an MRP tree for $\P$. Thus, also all of its refinements are MRP trees for $\P$, and as $MRC(\P)\preceq MRC(\P')$, we conclude that indeed $MRC(\P')$ is an MRP tree for $\P$. This concludes the proof.
\end{proof}

We will now turn our attention to non-contradiction, which shows a severe difference between MRP and MRC, as one of these methods turns out to be non-contradictory, while the other one is not. We start with the following negative result concerning MRP.

\begin{theorem} \label{thm:mrpnotnoncontr}
MRP is \emph{not} non-contradictory.
\end{theorem}

\begin{proof}
Consider profile $\P=(T_1,T_2)$ consisting of the two trees $T_1$ and $T_2$ from Figure \ref{fig:mrpnotnoncontr}. Moreover, consider the corresponding alignment of these two trees as given in Table \ref{tab:mrpnotnoncontr}. It can be easily seen that neither $T_1$ nor $T_2$ are MP trees for this alignment; in fact, the alignment has parsimony score 13 on tree $T_3$ from Figure \ref{fig:mrpnotnoncontr}, but 14 on both $T_1$  and $T_2$, respectively (cf. Table  \ref{tab:mrpnotnoncontr}). Moreover, we performed an exhaustive search over all phylogenetic trees with seven leaves using the computer algebra system Mathematica \cite{Mathematica} in order to verify that $T_3$ is even the unique MP tree for this alignment.\footnote{Note that there are $(2n-5)!!=(2\cdot 7 -5)!!=945$ binary phylogenetic trees on $7$ taxa \cite[Prop. 2.1.4]{Semple2003}, all of which need to be considered. As finding the maximum parsimony tree is unfortunately an NP-complete problem \cite{foulds_graham_1982}, there is no efficient way of solving it (except if $P=NP$), but for 945 trees, of course this task is still easily doable for a computer, which is why we decided to use Mathematica.} However, $T_3$ contains the split $\sigma=1234|567$, which is not contained in any of the input trees $T_1$ and $T_2$. As $T_1$ and $T_2$ are binary, this (by the Buneman theorem) necessarily means that $\sigma$ is incompatible with both trees, because no more split can be added to said trees. This completes the proof.

\begin{figure}
    \centering
   \includegraphics[scale=0.3]{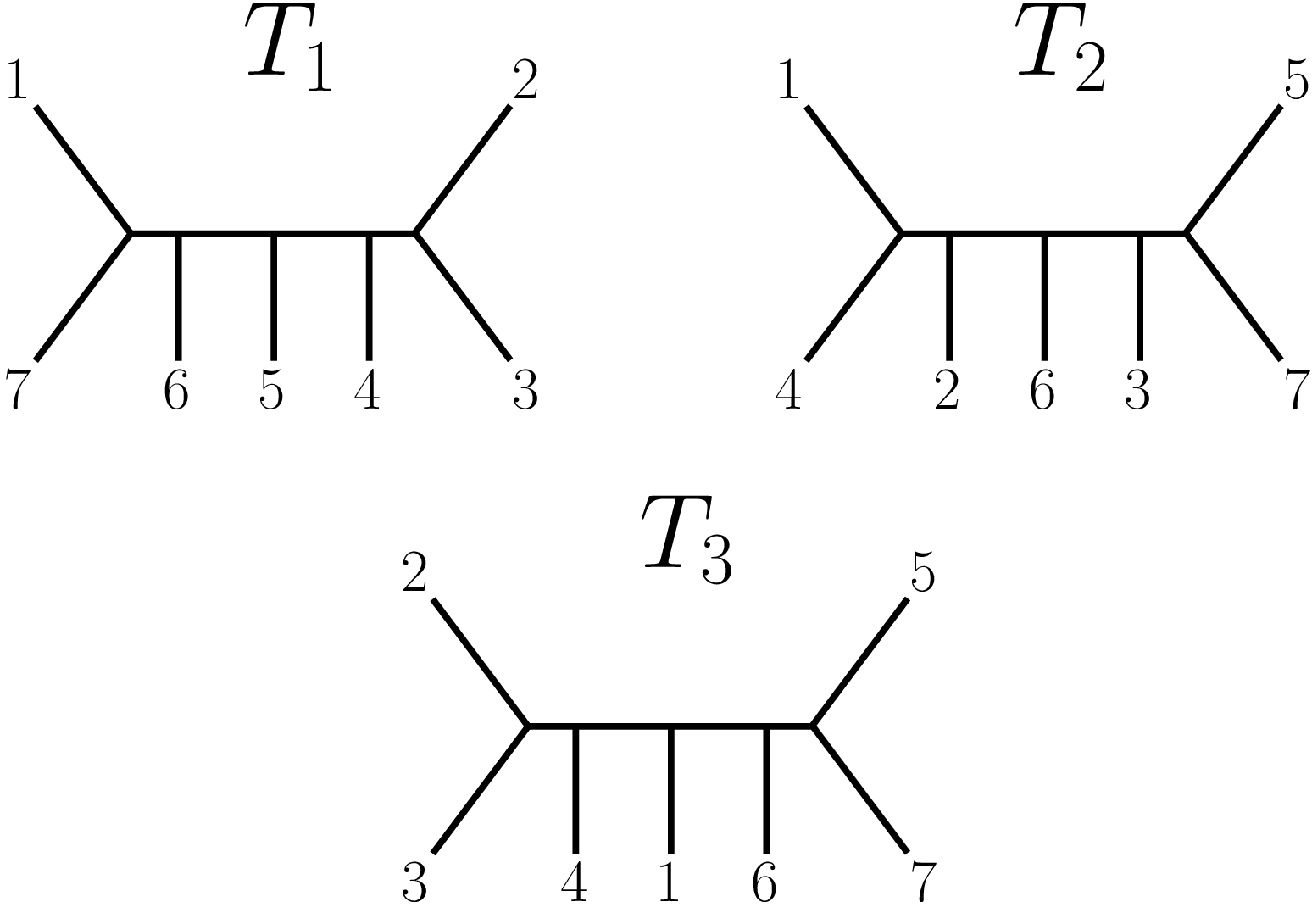}
    \caption{Trees $T_1$, $T_2$ and $T_3$ as needed for the proof of Theorem \ref{thm:mrpnotnoncontr}. The alignment for $\P=(T_1,T_2)$ is depicted in Table \ref{tab:mrpnotnoncontr}.}
    \label{fig:mrpnotnoncontr}
\end{figure}

\begin{table}\centering
\begin{tabular}{ll|llllllllll}
\cline{3-10}
                                                                                                                                      &       & \multicolumn{8}{c|}{alignment}                                          &                         &                                                                                                                                        \\ \cline{3-10}
                                                                                                                                      &       & 1 & 1 & 1 & \multicolumn{1}{l|}{1} & 1 & 1 & 1 & \multicolumn{1}{l|}{1} &                         &                                                                                                                                        \\
                                                                                                                                      &       & 0 & 0 & 0 & \multicolumn{1}{l|}{0} & 0 & 1 & 1 & \multicolumn{1}{l|}{1} &                         &                                                                                                                                        \\
                                                                                                                                      &       & 0 & 0 & 0 & \multicolumn{1}{l|}{0} & 0 & 0 & 0 & \multicolumn{1}{l|}{1} &                         &                                                                                                                                        \\
                                                                                                                                      &       & 0 & 0 & 0 & \multicolumn{1}{l|}{1} & 1 & 1 & 1 & \multicolumn{1}{l|}{1} &                         &                                                                                                                                        \\
                                                                                                                                      &       & 0 & 0 & 1 & \multicolumn{1}{l|}{1} & 0 & 0 & 0 & \multicolumn{1}{l|}{0} &                         &                                                                                                                                        \\
                                                                                                                                      &       & 0 & 1 & 1 & \multicolumn{1}{l|}{1} & 0 & 0 & 1 & \multicolumn{1}{l|}{1} &                         &                                                                                                                                        \\
                                                                                                                                      &       & 1 & 1 & 1 & \multicolumn{1}{l|}{1} & 0 & 0 & 0 & \multicolumn{1}{l|}{0} &                         &                                                                                                                                        \\ \hline
\multicolumn{1}{|l|}{\multirow{3}{*}{\rotatebox{90}{\mbox{$ps$}}}} & $T_1$ & 1 & 1 & 1 & 1                      & 2 & 3 & 3 & \multicolumn{1}{l|}{2} & \multicolumn{1}{l|}{14} & \multicolumn{1}{l|}{\multirow{3}{*}{\rotatebox{270}{\parbox{\textwidth}{total}}}} \\
\multicolumn{1}{|l|}{}                                                                                                                & $T_2$ & 2 & 3 & 3 & 2                      & 1 & 1 & 1 & \multicolumn{1}{l|}{1} & \multicolumn{1}{l|}{14} & \multicolumn{1}{l|}{}                                                                                                                  \\
\multicolumn{1}{|l|}{}                                                                                                                & $T_3$ & 2 & 2 & 1 & 1                      & 2 & 2 & 2 & \multicolumn{1}{l|}{1} & \multicolumn{1}{l|}{13} & \multicolumn{1}{l|}{}                                                                                                                  \\ \hline
\end{tabular}
\caption{The alignment resulting from \enquote{translating} trees $T_1$ and $T_2$  into binary characters and its corresponding parsimony scores on $T_1$, $T_2$ and $T_3$ from Figure \ref{fig:mrpnotnoncontr}.  }\label{tab:mrpnotnoncontr}
\end{table}

\end{proof}

Next, we present a positive result on MRC, before we turn our attention to \enquote{future-proofing} MRP and MRC in the sense presented in \cite{bryant2017can}, i.e., concerning the addition of new taxa. 

\begin{theorem}
MRC is co-Pareto and thus also non-contradictory.
\end{theorem}

\begin{proof}
By definition, the MRC tree is built by taking a maximum set of compatible splits of the input tree and combining them into a unique tree (which can be done using the Tree Popping algorithm mentioned in Section \ref{sec:basicConcepts}). This means that the MRC tree by definition only employs splits that occur in at least one input tree, i.e., MRC is co-Pareto. Note that this in particular implies that MRC is non-contradictory. This completes the proof.
\end{proof}

\subsubsection{``Future-proofing'' MRP and MRC concerning the addition of new taxa}

In order to complement the results of \cite{bryant2017can}, we briefly consider MRP and MRC in the light of adding more taxa (rather than more splits or clusters). This scenario is for instance relevant concerning the discovery of new species. The results stated in the following and the construction of the particular examples were strongly inspired by \citep{nonher}. As stated above, both MRP and MRC are strictly speaking not consensus methods as their output may be a set of trees rather than a single tree. However, we will here focus on a profile $\mathcal{P}$ for which both MRP and MRC produce a unique tree, namely $\widetilde{\mathcal{P}}=(\widetilde{T_1},\widetilde{T_1},\widetilde{T_2},\widetilde{T_2},\widetilde{T_2},\widetilde{T_3})$, where $\widetilde{T_1},\widetilde{T_2}, \widetilde{T_3}$ are as depicted in Figure \ref{fig_MRPMRCaddTaxon}. Now, we add taxon 5 to all trees in $\widetilde{\mathcal{P}}$ such that we get profile $\mathcal{P}= (T_1,T_1,T_2,T_3,T_4,T_5)$, where again $T_1, \ldots, T_5$ are as depicted in Figure \ref{fig_MRPMRCaddTaxon}. In the following, let the notation $T\setminus Y$ for a phylogenetic $X$-tree $T$ and a proper subset $Y$ of $X$ denote the phylogenetic $(X\setminus Y)$-tree that can be obtained from $T$ by deleting all leaves of $Y$ and suppressing all resulting degree-2 vertices as well as all newly occurring (i.e., unlabelled) leaves. Then, we have: 
$T_1\setminus\{5\}=\widetilde{T_1}$, $T_2\setminus\{5\}=T_3\setminus\{5\}=T_4\setminus\{5\}=\widetilde{T_2}$, and $T_5\setminus\{5\}=\widetilde{T_3}$. Translating all non-trivial splits of $\mathcal{P}$ and $\widetilde{\mathcal{P}}$ to characters leads to alignments $\mathcal{A}_\mathcal{P}$ and $\mathcal{A}_{\widetilde{\mathcal{P}}}$ as depicted in Figures \ref{fig_alignments1} and \ref{fig_alignments2}.

\begin{figure}
	\centering
	\includegraphics[scale=.35]{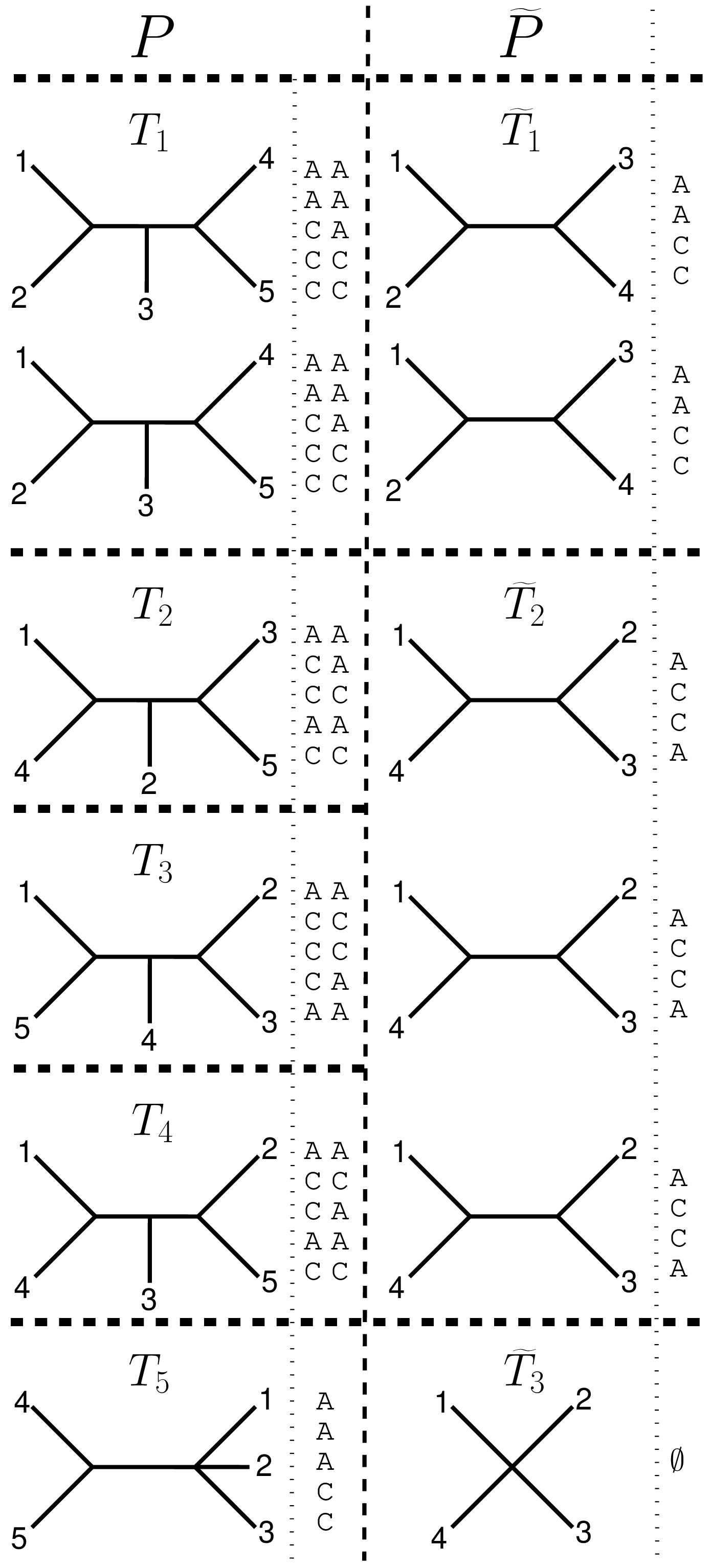}

	\caption{Profile $\widetilde{\mathcal{P}}$, which can be turned into profile $\mathcal{P}$ by adding information on taxon 5, leading to alignments $\mathcal{A}_{\widetilde{\mathcal{P}}}$ and $\mathcal{A}_\mathcal{P}$ as depicted in Figures \ref{fig_alignments1} and \ref{fig_alignments2}, respectively. $T_1$ is the unique MRP and MRC tree for $\mathcal{P}$, while $\widetilde{T_2}$ is the unique MRP and MRC tree for $\widetilde{T}$.
 	}
	\label{fig_MRPMRCaddTaxon}
\end{figure}

 We now argue that $T_1$ is both the unique MRP and MRC tree for $\mathcal{A}_\mathcal{P}$, and that $\widetilde{T_2}$ is both the unique MRP and MRC tree for $\mathcal{A}_{\widetilde{\mathcal{P}}}$. First note that resolving a tree more cannot increase the parsimony score, but it can potentially decrease it, i.e., for $T \preceq T'$ and a binary character $f$, we have $l(T',f)\leq l(T,f)$. This holds  because if multiple change edges are incident to the same vertex, replacing this vertex by a new edge might make it possible to have a change on the new edge rather than on multiple incident edges. Thus, when hunting for a tree minimizing the parsimony score, it is sufficient to consider binary trees. Now, as every binary character on four or five taxa requires at most 2 changes on any tree (it can always be realized with changes on the two edges incident to leaves that are not in the majority state), for a given character and a given tree with $n\leq 5$, there are only two possibilities: in case of compatibility, the parsimony score is 1, otherwise it is 2. So there are no different degrees of incompatibility (higher parsimony scores) for certain characters than for others, which is why it is clear that MRP and MRC necessarily coincide in this example: MRP will simply return the binary tree whose splits correspond to the characters occurring in total most frequently in $\mathcal{A}_{\P}$ and $\mathcal{A}_{\widetilde{\P}}$, which is precisely the MRC tree. While finding the MRC set in $\mathcal{A}_{\P}$ is a bit more involved, finding it in $\mathcal{A}_{\widetilde{\P}}$ is simple, as for four taxa, there are only three possible non-trivial splits, all of which are pairwise incompatible. Two of these splits' corresponding characters are present in  $\mathcal{A}_{\widetilde{\P}}$, and due to their incompatibility, MRC will simply pick the one that occurs most frequently, which is $ACCA$, leading to tree $\widetilde{T}_2$. \sout{Similarly, MRC will choose $T_1$ for $\mathcal{A}_{\P}$}.

We now briefly describe why $T_1$ is the unique MRC tree for $A_\mathcal{P}$. MRC requires finding a maximum compatible subset of characters displayed by the input trees. Namely, this is $3$ copies of $AAACC$, $2$ copies each of $AACCC$ and $ACCAC$, and $1$ copy each of $AACAC$, $ACCCA$, $ACCAA$ and $ACAAC$. Tree $T_1$ displays a subset of $5$ of these ($3$ copies of $AAACC$ and $2$ copies each of $AACCC$), so it suffices to show that this is the unique tree that displays a compatible subset of size $5$ or more. This will follow from several observations. Firstly, $ACCAC$ is incompatible with both $AAACC$ and $AACCC$, and secondly the pair $AAACC$ and $AACCC$ are not both compatible with any of the remaining characters. It follows that either our maximum subset consists of 3 copies of $AAACC$ and 2 copies of $AACCC$, or our subset contains all copies of at most one of $AAACC,AACCC$ and $ACCAC$ together with some subset of the remaining characters (of which there are one of each). Of these, $AAACC$ is only compatible with $ACCAA$ (a maximum size $4$ subset), $AACCC$ is only compatible with $AACAC$ (a maximum size $3$ subset) and while $ACCAC$ is compatible with $AACAC$, $ACCAA$ and $ACAAC$ (potentially a  size $5$ subset), $AACAC$ and $ACCAA$ are incompatible, leading again to a maximum size $4$ subset. It follows that the unique maximal subset of size $5$ is 3 copies of $AAACC$ and 2 copies of $AACCC$, which is uniquely displayed by $T_1$.

For the reasons explained above, these trees are then also the unique MRP trees in this case. 

However, as $\widetilde{T_2}$ is not a subtree of $T_1$, this shows that even in the case where MRP and MRC agree and additionally lead to a unique tree, additional taxa can change the predetermined relationships of the species in question. Thus, MRP and MRC are not `future-proof' in the sense of \cite{bryant2017can}.

\begin{figure}[ht]
\centering {\Large$\mathcal{A}_{\widetilde{P}}:$}
\parbox{1.5in}{ 
\begin{tabular}{cccccc}
1: &A & A & \bf{A} & \bf{A} & \bf{A} \\
 2:& A & A & \bf{C} & \bf{C} & \bf{C}  \\
 3: &C & C & \bf{C} & \bf{C} & \bf{C}\\
 4:& C & C & \bf{A} & \bf{A} &\bf{A}
\end{tabular}}
\caption{Alignment $\mathcal{A}_{\widetilde{\mathcal{P}}}$, which has both a unique MRP tree as well as a unique MRC tree: both of them equal $\widetilde{T_2}$ as depicted in Figure \ref{fig_MRPMRCaddTaxon}. The unique maximum set of compatible characters in $\mathcal{A}_{\widetilde{\mathcal{P}}}$ is hightlighted in bold.  }
\label{fig_alignments1}
\end{figure} 

\begin{figure}[ht]
\centering {\Large$\mathcal{A}_{\mathcal{P}}:$}
\parbox{3in}{ 
\begin{tabular}{cccccccccccc}
1: &\bf{A} & \bf{A} & \bf{A} & \bf{A} & A & A & A & A& A&A&\bf{A}\\
 2:& \bf{A} & \bf{A} & \bf{A} & \bf{A} & C & A & C& C& C&C&\bf{A}\\
 3: &\bf{C} &\bf{A} & \bf{C} & \bf{A} & C & C & C& C& C&A&\bf{A}\\
 4:& \bf{C} & \bf{C} & \bf{C} & \bf{C} & A & A & C& A &A&A&\bf{C}\\
 5: & \bf{C} & \bf{C} & \bf{C} & \bf{C} & C & C & A& A&C &C & \bf{C} 
\end{tabular}}
\caption{Alignment $\mathcal{A}_{\mathcal{P}}$, which has both a unique MRP tree as well as a unique MRC tree: both of them equal $T_1$ as depicted in Figure \ref{fig_MRPMRCaddTaxon}. The unique maximum set of compatible characters in $\mathcal{A}_{\mathcal{P}}$ is highlighted in bold.  }
\label{fig_alignments2}
\end{figure}

Next, we take a look at general consensus methods, i.e., we do no longer only consider established ones.

\subsection{General consensus methods}

\subsubsection{Monotonicity}

In this section, we gather some general properties of monotonic consensus methods. On the one hand, this will allow us to gain further insight into majority rule and strict consensus, two of the most frequently used methods in phylogenetics, as both are monotonic by Theorem \ref{thm_strictAndMLareRS}. On the other hand, however, general knowledge on monotonicity also allows for conclusions on what to hope for concerning future consensus methods -- what properties can you wish for if you are aiming at a monotonic method?

We have already seen that there are consensus methods that are not monotonic, such as the loose consensus. However, one can also construct  other examples of consensus methods which are not monotonic. For example, consider the following: Define the rather trivial method in which the star tree is always returned unless all trees in the profile are the same tree, in which case that tree is returned. Then this method is not monotonic. To see this, suppose the first profile consists of two identical non-star trees that are not fully resolved. If we then form the second profile by resolving exactly one of the trees more, the consensus becomes the star tree (and thus becomes coarser than it previously was). 

However, we start this section by showing the fundamental property that necessarily all monotonic and unanimous consensus methods refine the strict consensus, in the sense that such a consensus method applied to a given profile $\P$ will always return a refinement of $\mathtt{strict}(\P)$. Note that by Theorem \ref{thm_strictAndMLareRS} this applies, for instance, to majority rule.

\begin{theorem}
\label{RefStrict}
Let $\phi$ be any monotonic, unanimous consensus method and $\mathtt{strict}$ the strict consensus method. Let $\mathcal{P}$ be a profile of $k$ trees. Then we have: $\mathtt{strict}(\P) \preceq \phi(\P)$. In particular, if a cluster $c$ (or split $\sigma$ in the unrooted case) is induced by all trees in $\P$, then $c$ (or $\sigma$) is also induced by $\phi(\P)$.
\end{theorem}

\begin{proof} 
Let the strict consensus of $\P= (T_1,...,T_k )$ be some tree $T$, i.e., $T=\mathtt{strict}(\P)$. Then every tree in $\P$ contains each cluster  (or each split, respectively) of $T$. It follows that every tree in $\P$ is a refinement of $T$, and so given the profile $Q:=(T,...,T)$ consisting of $k$ copies of $T$ (where $k\geq 1$), we have that $T \preceq T_i$ for each $i$ in $\{1,...,k\}$, and as $\phi(Q) = T$ by unanimity and $\phi(Q) \preceq \phi(P)$ by monotonicity, we have $T \preceq \phi(P)$. This completes the proof.
\end{proof}

Theorem \ref{RefStrict} shows that the strict consensus is a coarse version of \emph{all} possible consensus trees induced by unanimous and monotonic methods. We now turn our attention to the loose consensus and prove that for a certain scenario, it is a refined version of \emph{all} possible consensus trees induced by unanimous and monotonic methods.

\begin{theorem} \label{lem_LooseIsFinest}
Let $\phi$ be any monotonic, unanimous consensus method, and let $\mathtt{loose}$ be the loose consensus method. Let $\P$ be a profile of $k$ trees such that each cluster (or each split, respectively) of every tree in $\P$ is compatible with all trees in $\P$. Then $\phi(\P) \preceq \mathtt{loose}(\P)$.
\end{theorem}

\begin{proof}
Let the loose consensus of $\P=(T_1,...,T_k)$ be the tree $T:=\mathtt{loose}(\P)$, noting that by definition of loose consensus, every cluster (or split, respectively) of each tree $T_i$ in $\P$ is contained in the clusters (or splits) of $T$ as all these clusters (or splits) are by assumption compatible with all trees in $\P$.  It follows that, given the profile $Q$ consisting of $k$ copies of $T$, we have that $T_i \preceq T$ for each $i$ in $\{1,...,k\}$, and as $\phi(Q) = T$ by unanimity and $\phi(P) \preceq \phi(Q)=T=\mathtt{loose}(\P)$ by monotonicity, we have $\phi(P) \preceq \mathtt{loose}(\P)$, which completes the proof.
\end{proof}

Note that Theorem \ref{lem_LooseIsFinest} is not quite as strong as Theorem \ref{RefStrict} in the following sense: While Theorem \ref{lem_LooseIsFinest} implies that the strict consensus tree is the coarsest refinement of any monotonic and unanimous consensus method's output tree regardless of any additional conditions, Theorem \ref{lem_LooseIsFinest} requires each cluster (or split) of a given profile $\P$ to be compatible with all trees in $\P$ in order to guarantee that the loose consensus refines all mentioned consensus methods' output trees. The fact that this condition is crucial and cannot be dropped is demonstrated by Example \ref{ex_LooseIsNotFinest}.

\begin{example} \label{ex_LooseIsNotFinest} For the unrooted case, consider again Figure \ref{fig:all4taxontrees}. Let $\P = (T_2,T_2,T_3)$, i.e., $\P$ employs two copies of $T_2$ and one copy of $T_3$. For the rooted case, we introduce a root on the inner edges of $T_2$ and $T_3$, respectively. It can easily be seen that in both cases, we have $\mathtt{MR}(\P)=T_2$, whereas $\mathtt{loose}(\P)$ is the star tree (i.e., $T_1$ in Figure \ref{fig:all4taxontrees}; in the rooted case the only inner vertex of $T_1$ is then the root). So in particular, we have $\mathtt{MR}(\P) \not\preceq \mathtt{loose}(\P)$. Since $\mathtt{MR}$ is unanimous (it is even regular, cf. \cite{bryant2017can}) and monotonic by Theorem \ref{thm_strictAndMLareRS}, this shows that the conditions of Theorem \ref{lem_LooseIsFinest} cannot be dropped. 
\end{example}

\par\vspace{0.5cm}
However, as we will now show, if there is a cluster (or split, respectively) that is incompatible with a cluster (or split) that is compatible with \emph{all} trees in a profile $\P$, then this cluster cannot be contained in the output tree of any regular and monotonic consensus method.

\begin{proposition} \label{prop_incompatible}
Suppose $\phi$ is a monotonic, unanimous  consensus method. Let $\P=(T_1,\ldots,T_k)$ be a profile of $k$ rooted (or unrooted) phylogenetic $X$-trees, and let $c_1$ and $c_2$ be two clusters (or $\sigma_1$ and $\sigma_2$ be two splits) such that $c_1$ (or $\sigma_1$) is compatible with all trees in $\P$, while $c_2$ (or $\sigma_2$) is not compatible with $c_1$ (or  $\sigma_1$, respectively). Then, $\phi(\P)$ does not contain $c_2$ (or $\sigma_2$).
\end{proposition}

\begin{proof} To see that the statement hold, suppose $\phi(\P)$ contained $c_2$ (or $\sigma_2$). We then refine all trees in $\P$ to contain $c_1$ (or $\sigma_1$), forming a new profile $\P'$ (this must be possible as  $c_1$ (or $\sigma_1$, respectively) is compatible with all trees in $\P$ by assumption). Then by definition of the strict consensus, $\mathtt{strict}(\P')$ contains $c_1$ (or $\sigma_1$). By Theorem \ref{RefStrict}, this implies $\phi(\P')$ contains $c_1$ (or $\sigma_1$). Thus, by monotonicity, $\phi(\P) \preceq \phi(\P')$. But this is impossible, since this would imply that $\phi(\P')$ contains both $c_1$ and $c_2$ (or $\sigma_1$ and $\sigma_2$), but these are incompatible by assumption. So this is a contradiction, which shows that $c_2$ (or $\sigma_2$) cannot be contained in $\phi(\P)$ to begin with. This completes the proof.
\end{proof}

As we will now show, the previous proposition implies that if there are two clusters (or splits) that are incompatible with one another, \emph{neither} one of them can be contained in the tree generated by a regular and monotonic consensus method if both of them are compatible with \emph{all} input trees. 

\begin{corollary}\label{cor_incompsplits}
Suppose $\phi$ is a regular, monotonic consensus method. Let $\P=(T_1,\ldots,T_k)$ be a profile of $k$ rooted (or unrooted) phylogenetic $X$-trees, and let $c_1$ and $c_2$ be two clusters (or $\sigma_1$ and $\sigma_2$ be two splits) compatible with all clusters (splits) of all trees in $\P$, but not with each other. Then $\phi(\P)$ contains neither $c_1$ nor $c_2$ (or neither $\sigma_1$ nor $\sigma_2$, respectively).
\end{corollary}

\begin{proof} As $c_1$ (or $\sigma_1$) is compatible with all trees in $\P$, by Proposition \ref{prop_incompatible}, $\phi(\P)$ cannot contain $c_2$ (or $\sigma_2$). Swapping the roles of $c_1$ and $c_2$, however, as now they are \emph{both} compatible with all trees in $\P$, yields that by the same argument, $\phi(\P)$ cannot contain $c_1$ (or $\sigma_1$). This completes the proof.
\end{proof}

Note that while the conditions of Corollary \ref{cor_incompsplits} may seem somewhat restrictive, they can actually quite easily be met. For instance, consider a profile containing multiple non-binary trees with the same unresolved cluster, e.g., let $\mathcal{P}=(T_1,\ldots,T_k$) be such that each $T_i$ contains subtree $T$ consisting of one inner vertex connected to its three only leaves $1$, $2$ and $3$. So each $T_i$ contains the cluster $c=\{1,2,3\}$, but none of the clusters $c_1=\{1,2\}$, $c_2=\{1,3\}$ and $c_3=\{2,3\}$. Any pair of the three latter clusters will then fulfill the requirements: These clusters are pairwise incompatible, but they are all compatible with all input trees. Thus, according to the above corollary, none of them can be contained in the output tree of a regular and monotonic consensus method.

So if two clusters (or splits) are incompatible with one another and if they are both compatible with all input trees, they cannot be contained in the output of any regular and monotonic consensus method. Again, note that by Theorem \ref{thm_strictAndMLareRS}, this for instance applies to majority rule and strict consensus, i.e., to two of most frequently used consensus methods. Another direct consequence of Proposition \ref{prop_incompatible} is the following corollary.

\begin{corollary} \label{cor:compT}
Suppose $\phi$ is a regular, monotonic consensus method. Let $\P=(T_1,\ldots,T_k)$ be a profile of $k$ rooted (or unrooted) phylogenetic $X$-trees, and let $T$ be some tree that is a refinement of all trees in $\P$. Then $\phi(\P)$ can consist only of clusters (or splits, respectively) compatible with $T$.
\end{corollary}

\begin{proof}
    Let $c$ be a cluster (or $\sigma$ be a split) not compatible with $T$. Then $c$  (or $\sigma)$ must be incompatible with at least one cluster $\widetilde{c} \in \mathcal{C}^*(T)$  (or split $\widetilde{\sigma} \in \Sigma^*(T)$). Then, as $T_i \preceq T$ for all $i=1,\ldots,k$ by assumption, $\widetilde{c} $ (or $\widetilde{\sigma}$) is compatible with $T_i$ for all $i=1,\ldots,k$. Thus, by Proposition \ref{prop_incompatible}, $c$ (or $\sigma$) cannot be contained in $\phi(\P)$. This completes the proof.
\end{proof}

{A final consequence of Proposition \ref{prop_incompatible} is the following corollary.}

\begin{corollary} \label{cor_reviewer}
{Suppose $\phi$ is a regular, monotonic consensus method and let $\P=(T_1,\ldots,T_k)$ be a profile of $k$ rooted (or unrooted) phylogenetic $X$-trees. Then any cluster (or split) in $\phi(P)$ is compatible with every cluster (or split) in $\mathtt{loose}(P)$.}
\end{corollary}

\begin{proof} Let $\phi$ be a regular, monotonic consensus method and let $\P=(T_1,\ldots,T_k)$ be a profile of $k$ rooted (or unrooted) phylogenetic $X$-trees. By definition, $\mathtt{loose}(\mathcal{P})$ contains only clusters (or splits) compatible with \emph{all} trees in $\mathcal{P}$. Thus, any cluster (or split) induced by $\mathtt{loose}(\mathcal{P})$ can play the role of $c_1$ (or $\sigma_1$) in Proposition \ref{prop_incompatible}, which implies that indeed, $\phi(\mathcal{P})$ cannot contain any cluster $c_2$ (or split $\sigma_2$) incompatible with this fixed choice of $c_1$ ($\sigma_1$). This completes the proof.
\end{proof}

As we have seen in Theorem \ref{thm_strictAndMLareRS}, majority rule and strict consensus are monotonic. Moreover, they are known to be regular \cite{bryant2017can}. We have also seen that not all methods are monotonic, and also not all methods are regular -- so is it possible that majority rule consensus and strict consensus are the \emph{only} consensus methods that have both properties? This would be mathematically interesting because it would imply that these two properties already give a  a full characterization for such methods. However, as we will later on see in Theorem \ref{thm:infinitelymany}, this is unfortunately not the case -- in fact, there are even infinitely many such methods. As we will see later on, the same is still true even if we enforce non-contradiction.

\subsubsection{Non-contradiction}

Before we turn our attention to the main result of this section, we will present a fundamental insight into consensus methods that are both non-contradictory and monotonic. In particular, we will show that under these circumstances, the output tree contains only clusters (or splits, respectively) that are already present in the input profile and thus is co-Pareto. Note that this is not automatically the case: non-contradiction requires only \emph{compatibility} with at least one input tree, not containment.

\begin{theorem} Let $\phi$ be a monotonic and non-contradictory consensus method. Let $\mathcal{P}=(T_1,\ldots,T_k)$ be a profile of phylogenetic $X$-trees. Then for every cluster $c$ (or split $\sigma$, respectively) in $\phi(\mathcal{P})$, there is at least one tree $T_i$ that displays $c$ (or $\sigma$){, i.e., $\phi$ is co-Pareto.}
\label{t:CIsPresent}
\end{theorem}

Before we can prove Theorem \ref{t:CIsPresent}, we need one more lemma.

\begin{lemma}\label{lem_refine} Let $T$ be a phylogenetic $X$-tree with split set $\Sigma(T)$ (or cluster set $\mathcal{C}(T)$ in the rooted case) and let $\sigma \not\in \Sigma(T)$ (or $c \not\in \mathcal{C}(T)$) be compatible with all splits in $\Sigma(T)$ (or all clusters in $\mathcal{C}(T)$, respectively). Then, there is a refinement $T'$ of $T$ which is incompatible with $\sigma$ (or $c$, respectively).
\end{lemma}

\begin{proof} We start with the rooted case. Let $T$, $\sigma$ and $c$ be as described in the theorem. We first consider the rooted case. Let $\widetilde{T}$ be the unique tree induced by $\mathcal{C}(T)\cup \{c\}$ (note that $\widetilde{T}$ must exist due to compatibility). Then $\widetilde{T}$ has an edge $e$ that $T$ does not have, namely the one inducing $c$. Contracting this edge turns $\widetilde{T}$ into $T$ and also leads to a vertex $v$, which is ancestral to all leaves of $c$. However, as $c$ is not a cluster of $T$, at least two subtrees $T_1$, $T_2$ rooted at children of $v$ in $T$ contain only leaves of $c$. But as $v$ does not induce $c$, there must be another pending subtree $T_3$ adjacent to $v$ which does not contain any leaf from $c$. Let $X_i$ denote the leaf set of $T_i$ for $i=1,2,3$, respectively. Then, let $c'=X_1\cup X_3$ and let $T'$ be the unique tree induced by $\mathcal{C}(T) \cup \{c'\}$. Then, $c$ and $c'$ are not compatible by construction and $T'$ is a refinement of $T$. This completes the proof for the rooted case.

The unrooted case follows analogously: We first add the split that is compatible with all splits in $T$ but not contained in its split set. Contracting this edge identifies a unique vertex in $T$ which, if we consider it to be the root, reduces the rest of the proof to the rooted case, which we have already shown. This completes the proof.  
\end{proof}

We are now finally in a position to prove Theorem \ref{t:CIsPresent}.

\begin{proof}[Proof of Theorem \ref{t:CIsPresent}]
Assume the theorem is false, i.e., assume that there are a monotonic and non-contradictory consensus method and a profile $\mathcal{P}=(T_1,\ldots,T_k)$ such that there is a cluster $c$ (or a split $\sigma$, respectively) in $\phi(\mathcal{P})$ not displayed by any $T_i$ (for $i=1,\ldots,k$). Then, as $\phi$ is non-contradictory, $c$ (or $\sigma$) is compatible with at least one tree $T_i$ (for $i\in \{1,\ldots,k$\}). By Lemma \ref{lem_refine} there is a refinement $T_i'$ of $T_i$ which is incompatible with $c$ (or $\sigma$). We replace $T_i$ by $T_i'$ in $\mathcal{P}$, and we repeat this procedure as long as there are trees compatible with $c$. We call the resulting profile $\mathcal{P}'$. As $\phi$ is monotonic, we know that $\phi(\mathcal{P}) \preceq \phi(\mathcal{P}')$. Thus, as $c$ (or $\sigma$) is contained in  $\phi(\mathcal{P})$, it is also contained in $\phi(\mathcal{P}')$. However, as $c$ is by construction incompatible with all trees in $\mathcal{P}'$, this contradicts the fact that $\phi$ is non-contradictory. This shows that the assumption was wrong and thus completes the proof.
\end{proof}

While we have already seen that the $\mathtt{MR}$ and $\mathtt{strict}$ consensus methods are regular, monotonic and non-contradictory (and all these are desirable properties of any consensus method, after all), whereas $\mathtt{loose}$ is not (as it lacks monotonicity), the question remains if this combination makes $\mathtt{strict}$ and $\mathtt{MR}$ unique or if there are other consensus methods -- biologically justified or not -- that also share all these properties. Indeed, Day and McMorris \cite[Section 1.2]{day2003axiomatic} divide major axiomatic consensus theory results into several categories, and a uniqueness result would be an excellent example of a result of the form \emph{"Consensus [method] $C$ is the unique rule $C$ if and only if $C$ has the desirable properties $X, Y,$ and $Z$"}. Unfortunately, this is not the case. The following theorem, which is the main result of this section, states that there are in fact infinitely many such consensus methods, so these properties are not unique to the mentioned ones.

\begin{theorem}\label{thm:infinitelymany}
There are infinitely many consensus methods {other than $\mathtt{strict}$ and $\mathtt{MR}$} that are regular, monotonic and non-contradictory.
\end{theorem}

\begin{proof} We consider both the rooted and the unrooted cases simultaneously. In order to do so, we choose $n\in \mathbb{N}_{\geq 4}$ and then fix a pair of trees $T_1$ and $T_2$ for each case. First, for the rooted case, consider the two phylogenetic $X$-trees $T_1$ and $T_2$ depicted in Figure \ref{fig_infinitelymany}, where $X=\{1,\ldots,n\}$. Then, for the unrooted case, consider the exact same trees, but suppress the root in each tree, i.e., replace the two edges incident to the degree-2 root vertex by a single edge. 

Note that the exact structure of the trees depends on whether $n=2m+2$ or $n=2m+3$ for $m \in \mathbb{N}$, i.e., whether $n$ is even or odd. Note that $T^*_{1\ldots m}$ and $T^*_{m+1\ldots 2m}$ are identical if leaf labels are disregarded, i.e., they correspond to the same tree shape, say $T^*$, but $T^*$ can be chosen to be \emph{any} rooted \emph{binary} tree (note that this implies that there are $WE(m)$ many choices for $T^*$, where $WE(m)$ denotes the $m^{th}$ Wedderburn-Etherington number \cite[Sequence A001190]{oeis}). Note that it is important that $T^*$ is binary as this implies that it cannot be refined any further. In particular, $T_2$ is already fully refined as all its subtrees are binary.

We now define a consensus method $\phi$ as follows: If a profile $\mathcal{P}$  contains some refinement $T_1'$ of $T_1$, i.e., $T_1\preceq T_1'$ (note that equality is possible) together with $T_2$, i.e., $\mathcal{P}=(T_1',T_2)$ or $\mathcal{P}=(T_2,T_1')$, then we define $\phi(\mathcal{P}):=T_2$.  Let $\pi$ be a permutation on $X$. For all phylogenetic $X$-trees $T$, let $\pi(T)$ denote the version of $T$ in which the leaves are permuted according to $\pi$. Then, we also define $\phi(\pi(T_1'),\pi(T_2)) := \pi(T_2)$. 

In all other cases, i.e., for all other input profiles $\widetilde{\mathcal{P}}$, we define $\phi(\widetilde{\mathcal{P}}):=\mathtt{strict}(\mathcal{\widetilde{P}})$. 

We now argue that $\phi$ is regular, monotonic and non-contradictory. Unanimity follows by the unanimity of $\mathtt{strict}$, and neutrality follows by $\mathtt{strict}$ for all profiles that are not $\mathcal{P}=(T_1',T_2)$ or $\mathcal{P}=(T_2,T_1')$, and for these profiles $\mathcal{P}$ it follows by our definition of $\phi$ as well as our consideration of permutations $\pi$. To check anonymity, note that $T_1'$ is either such that $T_1' \preceq T_2$ or such that $T_1'$ and $T_2$ have completely different tree shapes (in particular if $T_1$ gets refined in a way that does not form $T_2$, namely by combining leaf $2m+1$ with one of the copies of $T^*$ to form a new subtree). This shows that the only time when a permutation can modify $T_1'$ to give $T_2$ and $T_2$ to give $T_1'$ is when $T_1'$ and $T_2$ are isomorphic. However, in this case we actually have by unanimity that $\phi$ returns precisely this tree, so for anonymity, there remains nothing to show. Thus, $\phi$ is regular.

To see that $\phi$ is actually monotonic, let $\P^1=(T^1_1,\ldots,T^1_k)$ and $\P^2=(T^2_1,\ldots,T^2_k)$ such that $T^1_i\preceq T^2_i$ for all $i\in \{1,\ldots,k\}$. We need to show $\phi(\P^1)\preceq \phi(\P^2)$. Therefore, consider two cases: If $k=2$ and if $\P^1=(\pi(T_1'),\pi(T_2))$ (or $\P^1=(\pi(T_2),\pi(T_1'))$) for some $T_1'$ with $T_1\preceq T_1'$, by our definition of $\phi$, we have $\phi(\P^1) = \pi(T_2)$. However, as we know $T^1_i\preceq T^2_i$ for $i\in \{1,2\}$, we must have that $\P^2=(\pi(\widetilde{T}_1'),\pi(T_2))$ (or $\P^2=(\pi(T_2),\pi(\widetilde{T}_1'))$, respectively), where $\widetilde{T}_1'$ is a tree such that $T_1\preceq T_1'\preceq \widetilde{T}_1'$, i.e., it is a refinement of $T_1$. Concerning $T_2$, note that as $T_2$ is binary, any tree refining $T_2$ is necessarily identical to $T_2$. Thus, as $\P_2$ contains precisely $\pi(T_2)$ as well as a refinement of $\pi(T_1)$, by definition of $\phi$, we have $\phi(\P^2)=\pi(T_2)$. Thus, $\phi(\P^1)=\phi(\P^2)$, and thus trivially we have $\phi(\P^1)\preceq \phi(\P^2)$ as desired. 

In any other case, i.e., if $\P^1=(T^1_1,\ldots,T^1_k)$ and $\P^2=(T^2_1,\ldots,T^2_k)$ such that $T^1_i\preceq T^2_i$ for all $i\in \{1,\ldots,k\}$ but $\P^1\neq (\pi(T_1'),\pi(T_2))$ (and $\P^1\neq(\pi(T_2),\pi(T_1'))$) for all refinements $T_1'$ of $T_1$, then, by definition of $\phi$, have $\phi(\P^1)=\mathtt{strict}(\P^1)$. Now we consider two subcases.

\begin{itemize}\item Assume $\P^2\neq (\pi(T_1'),\pi(T_2))$ (and $\P^2\neq(\pi(T_2),\pi(T_1'))$) for all refinements $T_1'$ of $T_1$. In this case, by definition of $\phi$, we have $\phi(\P^2)=\mathtt{strict}(\P^2)$, and as $\P_1\preceq \P_2$ and as $\mathtt{strict}$ is monotonic by Theorem \ref{thm_strictAndMLareRS}, we must have $\phi(\P_1)\preceq \phi(\P_2)$ as desired.
\item The most involved case is the one in which $\P^1=(T^1_1,\ldots,T^1_k)$ and $\P^2=(T^2_1,\ldots,T^2_k)$ such that $T^1_i\preceq T^2_i$ for all $i\in \{1,\ldots,k\}$ but $\P^1\neq (\pi(T_1'),\pi(T_2))$ (and $\P^1\neq(\pi(T_2),\pi(T_1'))$) for all refinements $T_1'$ of $T_1$ and thus $\phi(\P^1)=\mathtt{strict}(\P^1)$, but $\P^2=(\pi(T_1'),\pi(T_2))$ for some refinement $T_1'$ of $T_1$, so  we have by definition of $\phi$ that $\phi(\P^2)=\pi(T_2)$. Moreover, this case implies $k=2$ and $\P^1=(T^1_1,T^1_2)$ with $T^1_1 \preceq \pi(T_1')$ and $T^1_2 \preceq \pi(T_2)$. We now show that $\phi(\mathtt{\P^1})=\mathtt{strict}(\P^1)$ contains only splits (or clusters, respectively) which are present in $\pi(T_2)$: Recall that all splits (or clusters)  present in $\mathtt{strict}(\P^1)$ must be present in both $T_1^1$ and $T^1_2$, so in particular the set of splits (clusters) of $\mathtt{strict}(\P^1)$ is a subset of the set of splits (clusters) of $T^1_2$. Moreover, note that by $T^1_2 \preceq \pi(T_2)$, all  clusters of $T^1_2$ are also present in $\pi(T_2)$. This implies that $\phi(\P^1)=\mathtt{strict}(\P^1) \preceq \pi(T_2) = \phi(\P^2)$ as desired.\end{itemize}

 Thus, in all cases we get $\phi(\P^1)\preceq \phi(\P^2)$, which in  summary shows that $\phi$ is monotonic.

Last, we need to show that $\phi$ is non-contradictory. To see this, recall that by Proposition \ref{prop:MRnoncontra}, $\mathtt{strict}$ is non-contradictory, and if a profile $\mathcal{P}$ consists of $\pi(T_1')$ and $\pi(T_2)$ for some permutation $\pi$, then $\phi$ returns $\pi(T_2)$. So in particular, all splits of $\phi((\pi(T_1'),\pi(T_2)))=\pi(T_2)$ are not only compatible with at least one of the input trees, but even contained in it, as $\pi(T_2)$ is contained in the input profile, which implies that $\phi$ is co-Pareto. 

So any consensus method $\phi$ of the type described here is regular, monotonic and non-contradictory. But we can construct infinitely many such methods by varying the trees depicted in Figure \ref{fig_infinitelymany}: We can, for instance, consider different values for $n$, so the cardinality of the set of such examples (and thus of such consensus methods) equals the cardinality of $\mathbb{N}$. This completes the proof.

\end{proof}

\begin{figure}
    \centering
	\includegraphics[scale=.3]{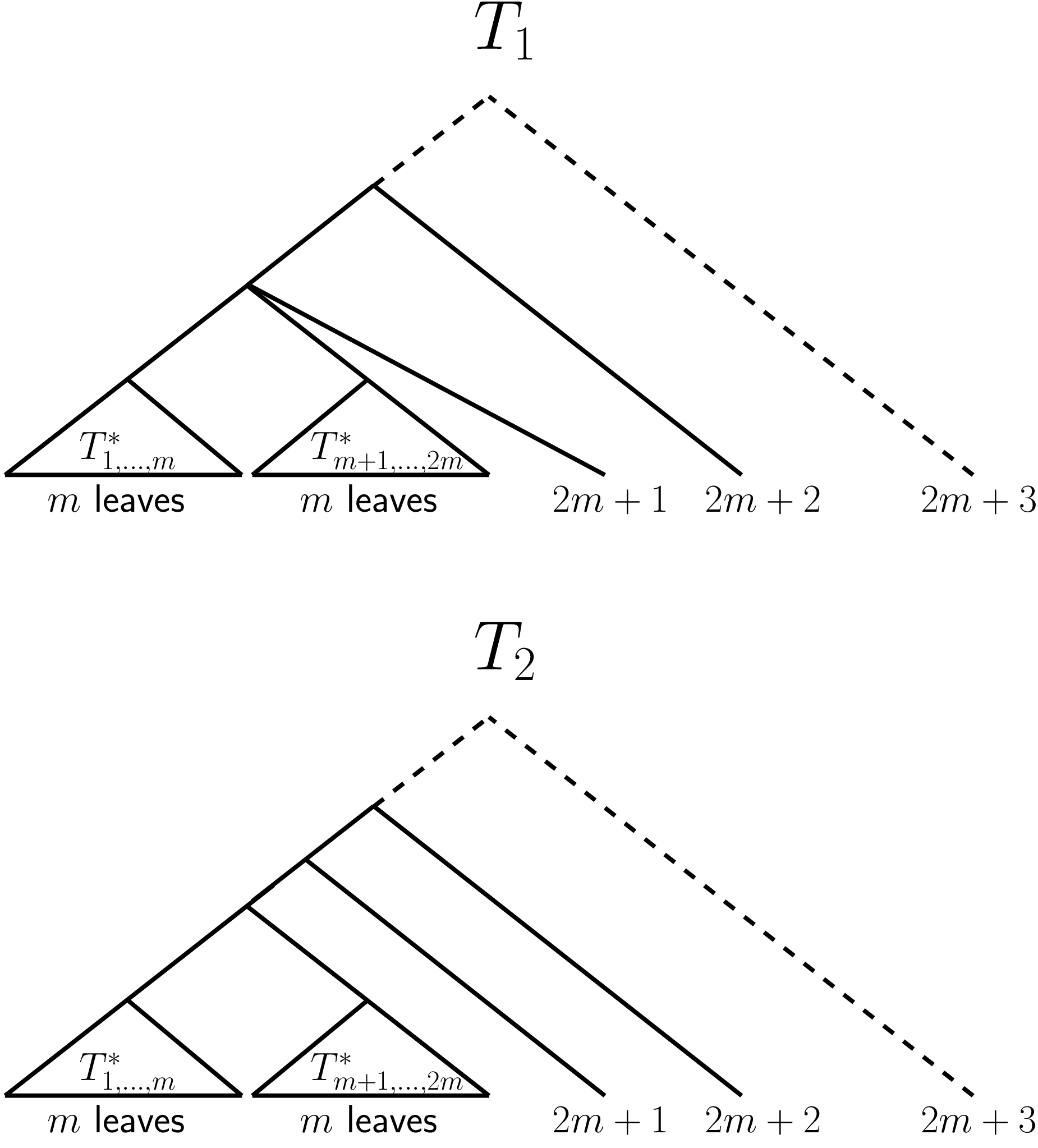}
	\caption{In the proof of Theorem \ref{thm:infinitelymany}, we have $\mathcal{P}=(T_1',T_2)$ on $n=2m+2$ or $n=2m+3$ leaves, respectively, where $T_1'$ is a refinement of $T_1$ as depicted here. $T^*_{1\ldots m}$ and $T^*_{m+1\ldots 2m}$ are both trees with the same binary tree shape $T^*$, but the leaves of $T^*_{1\ldots m}$ are bijectively labelled by $\{1,\ldots, m\}$ and the leaves of $T^*_{m+1\ldots 2m}$ are bijectively labelled by $\{m+1,\ldots, 2m\}$.  Note that the dashed parts of the trees only exist if $n$ is odd. 
 	}
	\label{fig_infinitelymany}
\end{figure}

\section{Discussion and outlook}
We investigated two  properties which are desirable for sensible phylogenetic consensus methods, namely monotonicity and non-contradiction, where the latter is a property that was newly introduced in the present manuscript. We analyzed seven established consensus methods and found that only two of them, namely strict consensus and majority-rule consensus, are monotonic, whereas the other five (loose consensus, Adams consensus, Aho consensus, MRP and MRC) are not. We further proved that strict, majority-rule and loose consensus as well as MRC are all non-contradictory, whereas Adams and Aho consensus as well as MRP are not. 

We also analyzed the implications that the properties of monotonicity and non-contradiction have on any consensus method, established or not, and, somewhat surprisingly, found that there are infinitely many consensus methods which are regular, monotonic and non-contradictory. As of the established consensus methods we analyzed, only two (namely strict and majority-rule) fit that description, it would be interesting to see if there are other biologically plausible consensus methods which have these properties. This is a possible area for future investigations.

\section*{Acknowledgements} Both authors wish to thank two anonymous reviewers for their helpful suggestions on a previous version of this manuscript. Moreover, MF wishes to thank the joint research project \textbf{\emph{DIG-IT!}}
supported by the European Social Fund (ESF), reference: ESF/14-BM-A55-0017/19, and the Ministry of Education, Science and Culture of Mecklenburg-Vorpommern, Germany. MF was also supported by the project ArtIGROW, which is a part of the WIR!-
Alliance \enquote{ArtIFARM – Artificial Intelligence in Farming}, and gratefully acknowledges the Bundesministerium f\"ur
Bildung und Forschung (Federal Ministry of Education and Research, FKZ: 03WIR4805) for financial support. MH thanks Prof. Dr. W. F. Martin, the Volkswagen Foundation 93\_046 grant and the ERC Advanced Grant No. 666053 for their support during this research. MH also thanks the Heinrich Heine University of D\"usseldorf, Germany, the University of Melbourne and Western Sydney University where some part of the work presented here was carried out.

\newpage
\section*{Declarations}

\begin{itemize}
    \item[]{\bfseries Funding:} The work of Mareike Fischer was supported by the European Social Fund (ESF), reference: ESF/14-BM-A55-0017/19, and the Ministry of Education, Science and Culture of Mecklenburg-Vorpommern, Germany. Moreover, Mareike Fischer's work was also supported by the project ArtIGROW, which is a part of the WIR!-Alliance \enquote{ArtIFARM – Artificial Intelligence in Farming}, and gratefully acknowledges the Bundesministerium f\"ur
Bildung und Forschung (Federal Ministry of Education and Research, FKZ: 03WIR4805). The work of Michael Hendriksen was supported by the Volkswagen Foundation 93\_046 grant and the ERC Advanced Grant No. 666053. 
     \item[]{\bfseries Conflicts of interest/Competing interests:} The authors declare that they have no competing interests and no conflicts of interest that are relevant to the contents of this article. 
     \item[]{\bfseries Availability of data and material:} Not applicable. 
     \item[]{\bfseries Code availability:} Not applicable. 
     \item[]{\bfseries Authors' contributions:} All authors contributed equally to the present manuscript. All mathematical analyses presented here were performed were performed by Mareike Fischer and Michael Hendriksen. Both authors wrote parts of the first draft. All authors read and approved the final manuscript.
     \item[]{\bfseries Ethics approval:} Not applicable.
     \item[]{\bfseries Consent to participate:} All authors give their consent to the publication of this manuscript and declare that they will participate in the publication process.
     \item[]{\bfseries Consent for publication:} All authors give their consent for publishing this manuscript.
    
\end{itemize}

\bibliographystyle{elsarticle-num} 
 \bibliography{References}

\begin{thebibliography}{10}
\expandafter\ifx\csname url\endcsname\relax
  \def\url#1{\texttt{#1}}\fi
\expandafter\ifx\csname urlprefix\endcsname\relax\def\urlprefix{URL }\fi
\expandafter\ifx\csname href\endcsname\relax
  \def\href#1#2{#2} \def\path#1{#1}\fi

\bibitem{oreilly}
J.~E. O’Reilly, P.~C.~J. Donoghue,
  \href{https://doi.org/10.1093/sysbio/syx086}{{The Efficacy of Consensus Tree
  Methods for Summarizing Phylogenetic Relationships from a Posterior Sample of
  Trees Estimated from Morphological Data}}, Systematic Biology 67~(2) (2017)
  354--362.
\newblock \href
  {http://arxiv.org/abs/https://academic.oup.com/sysbio/article-pdf/67/2/354/25092719/syx086.pdf}
  {\path{arXiv:https://academic.oup.com/sysbio/article-pdf/67/2/354/25092719/syx086.pdf}},
  \href {http://dx.doi.org/10.1093/sysbio/syx086}
  {\path{doi:10.1093/sysbio/syx086}}.
\newline\urlprefix\url{https://doi.org/10.1093/sysbio/syx086}

\bibitem{carling}
M.~Carling, R.~Brumfield, Integrating phylogenetic and population genetic
  analyses of multiple loci to test species divergence hypotheses in passerina
  buntings, Genetics 178 (2008) 363--377.
\newblock \href {http://dx.doi.org/10.1534/genetics.107.076422}
  {\path{doi:10.1534/genetics.107.076422}}.

\bibitem{geisler}
J.~Geisler, J.~Theodor, Hippopotamus and whale phylogeny, Nature 458 (2009)
  E1--4; discussion E5.
\newblock \href {http://dx.doi.org/10.1038/nature07776}
  {\path{doi:10.1038/nature07776}}.

\bibitem{day2003axiomatic}
W.~H. Day, F.~R. McMorris, Axiomatic consensus theory in group choice and
  biomathematics, SIAM, 2003.

\bibitem{arrow1950difficulty}
K.~J. Arrow, A difficulty in the concept of social welfare, Journal of
  political economy 58~(4) (1950) 328--346.

\bibitem{mirkin1975problem}
B.~Mirkin, On the problem of reconciling partitions, Quantitative sociology,
  international perspectives on mathematical and statistical modelling (1975)
  441--449.

\bibitem{mcmorris2008characterization}
F.~R. McMorris, R.~C. Powers, A characterization of majority rule for
  hierarchies, Journal of Classification 25 (2008) 153--158.

\bibitem{mcmorris}
F.~McMorris, D.~Neumann,
  \href{https://www.sciencedirect.com/science/article/pii/0165489683900999}{Consensus
  functions defined on trees}, Mathematical Social Sciences 4~(2) (1983)
  131--136.
\newblock \href
  {http://dx.doi.org/https://doi.org/10.1016/0165-4896(83)90099-9}
  {\path{doi:https://doi.org/10.1016/0165-4896(83)90099-9}}.
\newline\urlprefix\url{https://www.sciencedirect.com/science/article/pii/0165489683900999}

\bibitem{bryant2017can}
D.~Bryant, A.~Francis, M.~Steel, Can we “future-proof” consensus trees?,
  Systematic biology 66~(4) (2017) 611--619.

\bibitem{delucchi2019}
E.~Delucchi, L.~Hoessly, G.~Paolini,
  \href{https://doi.org/10.1093/sysbio/syz071}{Impossibility results on
  stability of phylogenetic consensus methods}, Systematic Biology\href
  {http://dx.doi.org/10.1093/sysbio/syz071} {\path{doi:10.1093/sysbio/syz071}}.
\newline\urlprefix\url{https://doi.org/10.1093/sysbio/syz071}

\bibitem{binindabook}
O.~Bininda-Emonds (Ed.), Phylogenetic Supertrees: Combining Information to
  Reveal the Tree of Life., Kluwer Academic Publishers, 2004.

\bibitem{bininda2003}
O.~Bininda-Emonds, {MRP} supertree construction in the consensus setting,
  Bioconsensus 61.
\newblock \href {http://dx.doi.org/10.1090/dimacs/061/16}
  {\path{doi:10.1090/dimacs/061/16}}.

\bibitem{newick}
J.~Felsenstein, J.~Archie, W.~Day, W.~Maddison, C.~Meacham, F.~Rohlf,
  D.~Swofford,
  \href{http://evolution.genetics.washington.edu/phylip/newicktree.html}{The
  {N}ewick tree format} (2000).
\newline\urlprefix\url{http://evolution.genetics.washington.edu/phylip/newicktree.html}

\bibitem{meacham}
C.~Meacham, Theoretical and computational considerations of the compatibility
  of qualitative taxonomic characters, in: J.~Felsenstein (Ed.), Numerical
  Taxonomy. NATO ASI Series (Series G: Ecological Sciences), Vol.~1, Springer,
  1983, pp. 304–--314.
\newblock \href
  {http://dx.doi.org/https://doi.org/10.1007/978-3-642-69024-2\_34}
  {\path{doi:https://doi.org/10.1007/978-3-642-69024-2\_34}}.

\bibitem{Semple2003}
C.~Semple, M.~Steel, Phylogenetics ({O}xford {L}ecture {S}eries in
  {M}athematics and {I}ts {A}pplications), Oxford University Press, 2003.

\bibitem{aho}
A.~V. Aho, Y.~Sagiv, T.~G. Szymanski, J.~D. Ullman,
  \href{https://doi.org/10.1137/0210030}{Inferring a tree from lowest common
  ancestors with an application to the optimization of relational expressions},
  {SIAM} Journal on Computing 10~(3) (1981) 405--421.
\newblock \href {http://dx.doi.org/10.1137/0210030}
  {\path{doi:10.1137/0210030}}.
\newline\urlprefix\url{https://doi.org/10.1137/0210030}

\bibitem{wilkinsoncotton}
M.~Wilkinson, J.~A. Cotton, F.-J. Lapointe, D.~Pisani,
  \href{http://www.jstor.org/stable/20143035}{Properties of supertree methods
  in the consensus setting}, Systematic Biology 56~(2) (2007) 330--337.
\newline\urlprefix\url{http://www.jstor.org/stable/20143035}

\bibitem{husondezulian}
D.~Huson, T.~Dezulian, T.~Klopper, M.~Steel, Phylogenetic super-networks from
  partial trees, IEEE/ACM Transactions on Computational Biology and
  Bioinformatics 1~(4) (2004) 151--158.
\newblock \href {http://dx.doi.org/10.1109/TCBB.2004.44}
  {\path{doi:10.1109/TCBB.2004.44}}.

\bibitem{margush1981consensus}
T.~Margush, F.~R. McMorris, Consensus n-trees, Bulletin of Mathematical Biology
  43~(2) (1981) 239--244.

\bibitem{bryant}
D.~Bryant, A classification of consensus methods for phylogenetics, in:
  Bioconsensus, 2001, pp. 163–--184.

\bibitem{levasseur}
C.~Levasseur, F.~Lapointe, Total evidence, average consensus and matrix
  representation with parsimony: What a difference distances make, Evolutionary
  bioinformatics online 2 (2006) 1--5.
\newblock \href {http://dx.doi.org/10.1177/117693430600200018}
  {\path{doi:10.1177/117693430600200018}}.

\bibitem{Fitch}
W.~M. Fitch, Toward defining the course of evolution: minimum change for a
  specific tree topology, Systematic Biology 20~(4) (1971) 406--416.
\newblock \href {http://dx.doi.org/10.1093/sysbio/20.4.406}
  {\path{doi:10.1093/sysbio/20.4.406}}.

\bibitem{Hartigan1973}
J.~Hartigan, Minimum mutation fits to a given tree, Biometrics 29~(1) (1973)
  53--65.

\bibitem{Mathematica}
{Wolfram Research Inc.}, \href{http://www.wolfram.com}{Mathematica 11.1}
  (2017).
\newline\urlprefix\url{http://www.wolfram.com}

\bibitem{foulds_graham_1982}
L.~Foulds, R.~Graham, The {S}teiner problem in phylogeny is {N}{P}-complete.,
  Advances in Applied Mathematics 3 (1982) 43--49.

\bibitem{nonher}
M.~Fischer, \href{https://doi.org/10.1007/s00285-011-0458-9}{Non-hereditary
  maximum parsimony trees}, Journal of Mathematical Biology 65~(2) (2012)
  293--308.
\newblock \href {http://dx.doi.org/10.1007/s00285-011-0458-9}
  {\path{doi:10.1007/s00285-011-0458-9}}.
\newline\urlprefix\url{https://doi.org/10.1007/s00285-011-0458-9}

\bibitem{oeis}
N.~Sloane, \href{https://oeis.org}{The {O}n-{L}ine {E}ncyclopedia of {I}nteger
  {S}equences {OEIS}} (2018).
\newline\urlprefix\url{https://oeis.org}

\end{thebibliography}

\end{document}